\documentclass[12pt]{article}
\usepackage[reqno]{amsmath}
\usepackage{amssymb,amsfonts,amsthm}
\usepackage{epsfig}

\def\cal{\mathcal}
\newcommand{\bC}{{\mathbb C}}

\newcommand{\vomega}{\boldsymbol{\omega}}
\newcommand{\vphi}{\boldsymbol{\phi}}
\newcommand{\vpsi}{\boldsymbol{\psi}}

\newcommand{\calO}{{\mathcal O}}

\newcommand{\interchange}{\iota}

\newcommand{\R}{{\mathbb R}}

\def\widebar{\overline}
\def\blamda{\widebar{\lambda}}

\def\realpart{\operatorname{\sf Re}}
\def\impart{\operatorname{\sf Im}}

\newcommand{\sn}{\operatorname{sn}}

\newcommand{\dn}{operatorname{dn}}
\newcommand{\cn}{operatorname{cn}}

\newcommand{\rd}{\mathrm{d}}            % Roman d for differential
\newcommand{\re}{\mathrm{e}}            % Roman e for exponential
\newcommand{\ri}{\mathrm{i}}            % Roman i for imaginary number
\newcommand{\rv}{\mathrm{v}}

\theoremstyle{plain}
\newtheorem{The}{Theorem}
\newtheorem{prop}[The]{Proposition}

\begin{document}
\title{Spectral stability of periodic NLS and CGL solutions}
\date{\today}

\author{{\sc T. Ivey\footnote{Supported by NSF grant DMS-0608587}
 and S. Lafortune\footnote{Supported by NSF grant DMS-0509622}}\\
 Dept. of Mathematics,\\
College of Charleston\\
Charleston, SC 29424, USA
}
\maketitle

\begin{abstract}
We consider periodic traveling wave solutions  to the focusing nonlinear
Schr\"odinger equation (NLS) that have been shown to 
persist when the NLS is perturbed to the complex Ginzburg-Landau equation (CGL). 
%Using squared Baker eigenfunctions, 
%we describe the spectrum of the linearized NLS. 
In particular, we show 
that these periodic traveling waves are spectrally stable solutions of NLS 
with respect to periodic perturbations. Furthermore, we use an argument 
based on the Fredholm alternative to find an instability criterion for the 
persisting solutions to CGL.
\end{abstract}

\section{Introduction}

The complex Ginzburg-Landau equation (CGL) is used in several contexts such as chemistry, fluid dynamics, and optics \cite{Ag91,Ko91,Ku84}. Furthermore, CGL was shown to be the generic equation modeling the slowly varying amplitudes of post-critical Rayleigh-B\'enard convection \cite{NeWh69}.
CGL is not in general an integrable equation. However, a successful strategy when it
comes to the study of existence and stability of solutions is to consider CGL as a perturbation
of the nonlinear Shr\"odinger equation (NLS). One can then take advantage of the fact 
that NLS is completely integrable. The drawback of that method is that the coefficients of the dissipative terms 
in CGL must remain small.

In the case of localized traveling wave solutions of CGL such as pulses and holes, this strategy has been implemented in several instances (see for example \cite{KaSa98,KaRu00,LeFa97}). 
In the case of periodic traveling waves, it has recently been shown that classes of solutions of NLS persist when considering CGL perturbations \cite{CrLeLu04a,CrLeLu04b,CrLeLu04c}. The main purpose of this paper is to study the stability of one of the classes of such solutions, namely the cnoidal wave solutions. We will restrict ourselves to the study of stability with respect to perturbations that have the same period as the waves which amounts to restricting the spatial domain of the solution to an interval of length equal to the period.

Our results are twofold.  First, we describe the spectrum of the linearization of the NLS about these cnoidal waves defined by (\ref{U0def}) below.  In particular, we show that, when considering perturbations that have the same periodicity as the wave, the cnoidal solutions are spectrally stable, i.e., the point spectrum of the linear operator arising from the linearization do not have eigenvalues with strictly positive real part.  In fact, we show that the location of the point spectrum is restricted to the imaginary axis. 
In addition, we obtain the continuous spectrum and show that it intersects the right side of the complex plane, showing that  cnoidal wave solutions (\ref{U0def}) are unstable when considered on the whole line. On the one hand, this is consistent with the fact that complex double points in the spectrum  of the Lax pair usually give rise to instabilities \cite{ErFoMc90,FoLe86,LiMc94}. On the other hand, it is also consistent with recent work \cite{CaDe06} in which (\ref{U0def}) is numerically shown to be unstable.

Our results concerning the spectrum of linearized NLS are summarized in Theorems 
\ref{the2} and \ref{the1} and illustrated in Figure \ref{fig3}.
The main tool we use is the fact that solutions to linearized NLS can be constructed
using squared AKNS eigenfunctions \cite{FoLe86,MOv}, for which explicit
formulas are available \cite{BBEIM,CI05}.
The main difficulty lies in the fact that one has to show that all the solutions of the eigenvalue problem are obtained that way; in fact, this is not true for the 
zero eigenvalue.  By using an alternate route to construct solutions,
and employing the technique of the Evans function (as defined in \cite{Ga97}), 
we show that the zero eigenvalue has geometric multiplicity.
 
Secondly, we use the knowledge of the spectrum of the linearized NLS to 
carry out a perturbative study of the spectrum of linearized CGL.
% identify an instability
%arising from the bifurcation of an eigenvalue from the origin when NLS is perturbed to CGL.
More precisely, we linearize about the persisting periodic solution of CGL, and we use the Evans function technique and the Fredholm alternative to identify
criteria for linear instability.  Under these conditions, eigenvalues of the linearized CGL bifurcate from the origin to the right side of the complex plane as NLS is perturbed to CGL. These results are summarized in Theorem \ref{newthm} and illustrated in Figures \ref{fig2} and \ref{fig7}.

\bigskip
The article is arranged as follows. In \S\ref{sec1},  we define the concept of Evans function for the linearized NLS about a periodic traveling wave. We then determine the spectrum of the linearized NLS about the cnoidal waves. In \S\ref{PS}, we describe the persisting solution of CGL and the corresponding linearization.  In \S\ref{near}, we use the concept of Evans function and  the Fredholm alternative to study eigenvalues emerging from the origin when NLS is perturbed.  The results of more detailed calculations are relegated to Appendix \ref{app:B} and the proof of Theorem \ref{the1} is given in Appendix \ref{app:nls}.

\subsection*{Acknowledgements}
The authors wish to thank Annalisa Calini, John Carter, and Joceline Lega for discussions and suggestions.

\section{Evans function for the linearized NLS about a periodic traveling wave}
\label{sec1}

\def\dn{\operatorname{dn}}
\def\cn{\operatorname{cn}}
\def\sn{\operatorname{sn}}

The NLS equation is written in system form as
\begin{equation}\label{NLS}
\begin{aligned}
\ri q_{1t}+q_{1xx}+2q_1^2 q_2& =0\\
-\ri q_{2t}+q_{2xx}+2q_1 q_2^2& =0.
\end{aligned}
\end{equation}
The reality condition $q_2=\overline{q_1}$ gives the focusing NLS for $q=q_1$.

We make a change of variables $q_1 = e^{-\ri \alpha t} u_1, q_2 = e^{\ri \alpha t} u_2$,
for a real constant $\alpha$, and obtain the following system:
\begin{equation}\label{NLSc}
\begin{aligned}
\ri u_{1t}+\alpha u_1+u_{1xx}+2u_1^2 u_2& =0\\
-\ri u_{2t}-\alpha u_2+u_{2xx}+2u_1 u_2^2& =0.
\end{aligned}
\end{equation}
This system has several $T$-periodic traveling wave solutions of the
form $u_1=\overline{u_2}=U_0(x-ct)$ where $U_0$ is expressible in terms of elliptic functions.
However, these
only persist as $T$-periodic solutions of the complex Ginzburg-Landau (CGL) equation when $c=0$ \cite{CrLeLu04a,CrLeLu04b,CrLeLu04c}. Thus, we limit our attention
here to the real solution of \eqref{NLSc} given by
\begin{equation}
\label{U0def}
u_1(x,t)={u_2}(x,t)=U_0(x), \qquad U_0(x) =  \delta k \cn(\delta\,x;k),
\end{equation}
with $\alpha=\delta^2(1-2 k^2)$ where $k$ is the elliptic modulus, $0<k<1$.
The period  of this solution is
$T = 4K/\delta$, where $K=K(k)$ is the complete elliptic integral of the first kind.

We then linearize the system (\ref{NLSc}) about the solution (\ref{U0def}).
Substituting $u_1=U_0+\rv_1$,
 $u_2=U_0+\rv_2$ in (\ref{NLSc}) and discarding higher-order terms in the $\rv$'s, we get
\begin{equation}\label{vLNS}
\begin{aligned}
\ri \rv_{1t}+\alpha \rv_1 + \rv_{1xx}+ 4 U_0^2 \rv_1+ 2U_0^2 \rv_2 &=0 \\
-\ri \rv_{2t}+\alpha \rv_2 + \rv_{2xx}+ 4 U_0^2 \rv_2+ 2U_0^2 \rv_1 &= 0.
\end{aligned}
\end{equation}
%here is how the Evans function is defined.
We substitute into \eqref{vLNS} the ansatz
\begin{equation}
\rv_1 = \re^{\ell t} w_1(x), \qquad \rv_2 = \re^{\ell t} w_2(x),
\end{equation}
where $w_1$ and $w_2$ are assumed to be independent of time $t$,
obtaining the following  ODE system for $w_1,w_2$
\[
\begin{aligned}
\ri \ell w_{1}+\alpha w_1 + w_{1xx}+ 4 U_0^2w_1+ 2 U_0^2 w_2 &=0,\\
-\ri \ell w_{2}+\alpha w_2 + w_{2xx}+ 4 U_0^2w_2+ 2U_0^2 w_1 &=0,
\end{aligned}
\]
which can be written as the eigenvalue problem
\begin{equation}
\mathcal{L}_0\vomega = \ell \vomega,\;\;\;\vomega=\left(w_1,w_2\right)^T
\label{eigenvalueproblem}
\end{equation}
where
\begin{equation}
\label{Ldef}
{\mathcal L}_0=J\left[\frac{\rd^2}{\rd x^2}+ \begin{pmatrix}4U_0^2+\alpha & 2 U_0^2 \\
2 U_0^2 & 4U_0^2+\alpha
\end{pmatrix}\right],\;\;\;J=\begin{pmatrix} \ri & 0 \\ 0 & -\ri \end{pmatrix}.
\end{equation}
Since we are interested in periodic perturbations, we define the point spectrum of the eigenvalue problem (\ref{eigenvalueproblem}) to consist of values of $\ell$ for which $\vomega$ is an element of the space $\left(L_2(T)\right)^2$ of two-dimensional vectors with components in the set of $T$-periodic square integrable functions.

In order to define the Evans function, we rewrite the ODE system (\ref{eigenvalueproblem}) as the following first-order linear system with $T$-periodic coefficients:
\begin{equation}\label{fourbysys0}
\dfrac{\rd}{\rd x} {\bf{w}}
={\cal A}(\ell,x)\,{\bf{w}},\;\;{\bf{w}}=
\left( w_1,\; w_2 ,\; w_1' ,\;w_2' \right)^T,
\end{equation}
where ${\cal A}(\ell,x)$ is the $4\times 4$ matrix given by
\begin{equation}
{\cal A}(\ell,x)=
\begin{pmatrix}
\label{A}
0& 0 & 1 & 0 \\
0 & 0 & 0 & 1 \\
-(\alpha+\ri \ell + 4 U_0^2) & -2U_0^2 & 0 & 0\\
-2 U_0^2 &  -(\alpha -\ri \ell + 4U_0^2) & 0 & 0
\end{pmatrix}.
\end{equation}
  Let $\Phi(x;\ell)$ be a fundamental
matrix for this system, with $\Phi(0;\ell)=I$.  Then the Evans function  \cite{Ga97} is defined to be
\begin{equation}
\label{EvansDef1}
M(\kappa,\ell) = \det\left(\Phi(T;\ell) - e^{\ri \kappa} I\right).
\end{equation}
If $M(\kappa,\ell) = 0$ for some $\kappa\in \R$ then the system (\ref{fourbysys0}), or equivalently (\ref{eigenvalueproblem}), has a bounded solution. The zero set of $M(\kappa,\ell)$ thus corresponds to 
the continuous spectrum of $\mathcal{L}_0$.
However, we restrict the eigenvalue problem (\ref{eigenvalueproblem}) to $T$-periodic square integrable functions  which corresponds to setting $\kappa=0$. We thus define
\begin{equation}
\label{EvansDef}
M_0(\ell) \equiv M(0,\ell).
\end{equation}
The zero set of $M_0(\ell)$ corresponds to the point spectrum of the eigenvalue problem (\ref{eigenvalueproblem}) and the multiplicity of the zero of $M_0$ corresponds to the algebraic multiplicity of each eigenvalue \cite{Ga97}.

For all but a few exceptional values of $\ell$, the Baker eigenfunctions associated to periodic
traveling wave solutions of NLS may be used to construct a fundamental
solution matrix for the system \eqref{fourbysys0}.  (Among the exceptions
is the important case $\ell=0$, which is treated in \S\ref{eps=0}.)
Then the Floquet
discriminant of the AKNS system may be used to determine the zeros
of the Evans function. The following theorem, which establishes the spectral
stability of the NLS solution  (\ref{U0def}) with respect to $T$-periodic 
perturbations, is proven in Appendix \ref{app:nls}.
\begin{The}
\label{the2}
The zero set of the Evans function (\ref{EvansDef}) which corresponds to the point spectrum 
of $\mathcal{L}_0$ defined in (\ref{eigenvalueproblem}) consists of a countably infinite number of points on the imaginary axis in the $\ell$-plane. 
With the possible exception of four values, the locations of these points is 
determined by the formula $\ell=4\ri \mu$ where 
\begin{equation}\label{mupoly}\mu^2=(\lambda^2-\lambda_1^2)(\lambda^2-\overline{\lambda_1}^2),
\end{equation}
and $\lambda$ is a value for which the Floquet discriminant 
$\Delta(\lambda)$ of the AKNS system equals $\pm 2$ or $0$.
\end{The}

Note that the branch point $\lambda_1$ is related to the solution (\ref{U0def}) by $\lambda_1 / |\lambda_1| = -k' + \ri k$ and $\delta=2\left|\lambda_1\right|$, where $k'=\sqrt{1-k^2}$ is
the complementary modulus
and the Floquet discriminant is defined in \eqref{deltaformula}.   
The exceptional values are those for which \eqref{mupoly}, as 
a polynomial equation for $\lambda$, has multiple roots.  Note that these
all correspond to imaginary values of $\ell$.

Figure \ref{fig3} shows the location of the spectrum in the $\ell$-plane for $\delta=2$ and three different 
values of $k$.
In addition, the continuous spectrum is also obtained in Appendix \ref{app:nls}. It consists of the imaginary axis and a figure eight centered at the origin (see Figure \ref{fig3}). Thus, this shows that
solution (\ref{U0def}) is unstable with respect to bounded perturbations and is consistent with the fact that complex double points in the spectrum  of the Lax pair usually give rise to instabilities \cite{ErFoMc90,FoLe86,LiMc94}. 

\section{The perturbed setting}
\label{PS}

\subsection{Periodic traveling wave solutions to CGL}
\label{PS1}

The CGL perturbation of the NLS system has the following form:
\begin{equation}\label{CGL}
\begin{aligned}
\ri q_{1t}+q_{1xx}+2q_1^2 q_2 =
\epsilon & [\ri r q_1 +\ri q_{1xx}-2\ri s q_1^2 q_2], \\
-\ri q_{2t}+q_{2xx}+2q_1 q_2^2 =
\epsilon & [-\ri r q_2 -\ri q_{2xx}+2\ri s q_1 q_2^2] .
\end{aligned}
\end{equation}
Here, $\epsilon>0$, and $r,s$ are real parameters.
Again, $q_2=\overline{q_1}$ is the reality condition for the focusing case.

We again make the change of variables $q_1 = e^{-\ri \alpha t} u_1, q_2 = e^{\ri \alpha t} u_2$,
and obtain the following system:
\begin{equation}\label{CGLc}
\begin{aligned}
\ri u_{1t}+\alpha u_1+u_{1xx}+2u_1^2 u_2& =\epsilon [\ri r u_1 +\ri u_{1xx}-2\ri s u_1^2 u_2]\\
-\ri u_{2t}-\alpha u_2+u_{2xx}+2u_1 u_2^2& =\epsilon [-\ri r u_2 -\ri u_{2xx}+2\ri s u_1 u_2^2] ,
\end{aligned}
\end{equation}
with the reality condition $u_2=\overline{u_1}=u$.

The solution to NLS given in (\ref{U0def})
persists as a $T$-periodic solution of the CGL (\ref{CGLc}) if $r$ and $s$ satisfy a certain
algebraic condition. More precisely, it was shown in \cite{CrLeLu04a,CrLeLu04b,CrLeLu04c} that for small $\epsilon>0$, the CGL equation \eqref{CGLc} has a solution which is $T$-periodic  in $x$ of the form
 \begin{equation}
 \label{Udef}
 u(x,t)= U(x)\;\;\; {\mbox{with}}\;\;\; U(x)=U_0(x)+\epsilon U_1(x)+{\cal{O}}(\epsilon^2),
 \end{equation}
  where $U_0$ is the solution of NLS given in (\ref{U0def}), provided that $r$ and $s$ satisfy
\begin{equation}
\label{req}
 r={\frac {{\delta}^{2}
 \left(E\,(1 +4\,s)\,(2\,k^2-1)+K\, k'^2\,(1+4\,s-6\,s{k}^{2}) \right) }
{3\,(E- K\, k'^{2})}}
\;\;{\mbox{and}}\;\;r\geq\frac{\delta^2 \pi^2}{K^2},
\end{equation}
where $E=E(k)$
is the complete elliptic integral of the second kind.

In order to investigate stability properties of the CGL solution $U$, we will need the expression giving $U_1$. To find it, we substitute $u_1=U$ and $u_2=\overline{U}$ into (\ref{CGLc}). At first order in $\epsilon$, one finds that the vector
 ${\bf{U}}=\left(U_1,\;\overline{U_1}, \;U'_1,\;{\overline{U_1}'}\right)^T$ satisfies
  a non-homogeneous linear dynamical system. This system is given by
 \begin{equation}
 \label{NHS}
 \dfrac{\rd}{\rd x} {\bf{U}}={\cal A}(0,x)\,{\bf{U}}+{\bf{N}},
\end{equation}
where ${\cal A}(0,x)$ is obtained by setting $\ell=0$ in the matrix given in (\ref{A}), and the non-homogenous term ${\bf{N}}$ is given by
$$
{\bf{N}}=
 \begin{bmatrix}
 0\\0\\\ri\,r\,U_0+\ri\,U_0''-2\,\ri\,s\,U_0^3\\-\ri\,r\,U_0-\ri\,U_0''+2\,\ri\,s\,U_0^3
 \end{bmatrix}.
$$

The homogeneous system corresponding to (\ref{NHS}) is given by
\begin{equation}
 \dfrac{\rd}{\rd x} {\bf{U}}={\cal A}(0,x)\,{\bf{U}}.
 \label{Hom}
 \end{equation}
It turns out that (\ref{Hom}) can be solved explicitly. To do so, one uses two symmetries of CGL (\ref{CGLc}): translation of the variable $x$ and phase change $q_i\rightarrow q_ie^{i\theta}$ for any real $\theta$. The presence of these two symmetries implies the following $T$-periodic solutions to (\ref{Hom}):
\begin{equation}
\label{A0sol}
{\bf{V}}^h_1= \begin{bmatrix}
 U_0'\\U_0'\\U_0''\\U_0''
 \end{bmatrix},\;\;\;
 {\bf{V}}^h_2=
 \begin{bmatrix}
 U_0\\-U_0\\U_0'\\-U_0'
 \end{bmatrix}.
 \end{equation}
 One can then use these two solutions to perform reduction of order on the homogeneous system. The two-dimensional homogeneous linear dynamical system for two unknowns obtained that way turns out to be diagonal and thus can be solved by quadrature. The expressions for the other two solutions ${\bf{V}}^h_3$ and  ${\bf{V}}^h_4$ are rather complicated. They are given in Appendix \ref{app:B}. 
 No linear combination of ${\bf{V}}^h_3$ and  ${\bf{V}}^h_4$ is periodic, implying that (\ref{Hom}) has only a two-dimensional space of $T$-periodic solutions generated by ${\bf{V}}^h_1$ and  ${\bf{V}}^h_2$. Once four linearly independent solutions to \eqref{Hom} are known, one can then use the variation of parameters method to find a
particular solution ${\bf{U}}$ to (\ref{NHS}). Requiring this solution to be $T$-periodic leads to the algebraic relations (\ref{req}) between the constants  $k$, $r$, $\delta$, and $s$.
 The first component of ${\bf U}$ is then given by
\begin{equation}
\label{U1def}
\begin{aligned}
 U_1&=\ri\,\delta k (s+1)\cn(\delta x)\,\left[-\frac{2}{3}\ln \left(\theta_4\left(\frac{\pi\,\delta\,x}{2\,K}\right)\right)\right.\\
 &+\left.\frac{K\,\left(2\dn(\delta\,x)\,\sn(\delta\,x)\,Z(\delta\,x)-Z^2(\delta\,x)\,\cn(\delta\,x)+\cn^3(\delta\,x)\,k ^2\right) }{3\,\cn(\delta\,x)(E+(k ^2-1)\,K)}\right],
 \end{aligned}
 \end{equation}
where $Z$ is the Jacobi zeta function
and $\theta_4$ is a Jacobi theta function.  (For definitions and conventions
for these and other relatives of the elliptic functions, we refer the reader to
\cite{BF}.)

Note that any linear combination of ${\bf{V}}_1^h$
or ${\bf V}_2^h$ can be added to $\mathbf{U}$ without losing the periodicity property.
This latter fact can be explained by the presence of the two symmetries of CGL mentioned above: adding a multiple of the top component of ${\bf{V}}_1^h$ corresponds to adding translation in $x$ to the CGL perturbation, while adding a multiple of
the top component of ${\bf{V}}_2^h$ corresponds to adding a phase change.

Now that we know more about the dependence of $U(x)$ in terms of $\epsilon$, we can study its spectral stability as a solution of CGL

 \subsection{Linearization of CGL}

 We linearize the system of equations (\ref{CGLc}) about the stationary $T$-periodic solution
$u_1(x,t)=U(x)$, $u_2(x,t) = \overline{U}(x)$, where $U$ is given in (\ref{Udef}).  We do this
in a similar way as it is done for NLS in \S \ref{sec1}.
Setting
$u_1=U+\rv_1$, $u_2=\overline{U}+\rv_2$ and discarding all higher-order terms in the
$\rv_i$'s, we obtain
\begin{equation}\label{vLCGL}
\begin{aligned}
\ri \rv_{1t}+\alpha \rv_1 + \rv_{1xx}+ 4 |U|^2\rv_1+ 2U^2\rv_2 &=\epsilon[\ri r \rv_1 + \ri \rv_{1xx}-4 \ri s |U|^2\rv_1-2\ri s U^2 \rv_2], \\
-\ri \rv_{2t}+\alpha \rv_2 + \rv_{1xx}+ 4 |U|^2\rv_2+ 2\overline{U}^2\rv_1 &=\epsilon[-\ri r \rv_2 - \ri \rv_{2xx}+ 4 \ri s |U|^2 \rv_2+ 2 \ri s \overline{U}^2 \rv_1].
\end{aligned}
\end{equation}

Next, by analogy with \S\ref{sec1} we substitute into \eqref{vLCGL} the ansatz
\begin{equation}
\rv_1(x,t) = e^{\ell t} w_1(x), \qquad \rv_2(x,t) = e^{\ell t} w_2(x).
\end{equation}
This results in the following ODE system for $w_1,w_2$
\[
\begin{aligned}
\ri \ell w_{1}+\alpha w_1 + w_{1xx}+ 4 |U|^2w_1+ 2U^2w_2 &=\epsilon[\ri r w_1 + \ri w_{1xx}-4 \ri s |U|^2w_1-2\ri s U^2 w_2] \\
-\ri \ell w_{2}+\alpha w_2 + w_{1xx}+ 4 |U|^2w_2+ 2\overline{U}^2w_1 &=\epsilon[-\ri r w_2 - \ri w_{2xx}+ 4 \ri s |U|^2w_2 + 2 \ri s \overline{U}^2 w_1],
\end{aligned}
\]
which can be rewritten as the eigenvalue problem
\begin{equation}
\mathcal{L}\vomega = \ell \vomega,\;\;\;\vomega:=\left(w_1,w_2\right)^T.\;\;
\label{CGLeigenvalueproblem}
\end{equation}
If we expand $\mathcal{L}=\mathcal{L}_0+\epsilon\mathcal{L}_1+\epsilon^2\mathcal{L}_2+{\mathcal{O}}(\epsilon^3)$, then  $\mathcal{L}_0$ is the operator arising from the linearization of NLS given in (\ref{Ldef}) and
$\mathcal{L}_1$ is given by
$$
\mathcal{L}_1=I_2\frac{\rd^2}{\rd x^2}+ \begin{pmatrix}r-4sU_0^2 & 2 U_0\left( 2\ri U_1-sU_0 \right) \\
2 U_0\left( 2\ri U_1-sU_0 \right) & r-4sU_0^2
\end{pmatrix},\;\;\;I_2=\begin{pmatrix} 1 & 0 \\ 0 & 1 \end{pmatrix}.
$$
Since we are interested in periodic perturbations, we restrict the eigenvalue problem (\ref{CGLeigenvalueproblem}) to the space $\left(L_2(T)\right)^2$ of $T$-periodic square integrable functions.

%We rewrite (\ref{CGLeigenvalueproblem}) as a first-order linear system,
%with $T$-periodic coefficients which involve $\ell$:
%\begin{equation}\label{fourbysys}
%\dfrac{d}{dx} \begin{bmatrix} w_1 \\ w_2 \\ (1-\ri \epsilon) w_1' \\ (1+\ri \epsilon) w_2' \end{bmatrix}
%= \begin{pmatrix}
%0& 0 & 1 & 0 \\
%0 & 0 & 0 & 1 \\
%-[\alpha+\ri (\ell-\epsilon r) + 4 |U|^2(1+\ri \epsilon s)] & -2(1+\ri \epsilon s)U^2 & 0 & 0\\
%-2(1-\ri \epsilon s)\overline{U}^2 &  -[\alpha -\ri (\ell-\epsilon r) + 4|U|^2(1-\ri \epsilon s)] & 0 & 0
%\end{pmatrix}
%\begin{bmatrix} w_1 \\ w_2 \\ w_1' \\ w_2' \end{bmatrix}.
%\end{equation}

 \section{Eigenvalues near the origin}
 \label{near}

In this section, we study stability properties of $U(x)$ as a solution to CGL by studying the spectrum of the operator $\mathcal{L}$. We use a perturbative argument to study the eigenvalues of $\mathcal{L}$ near $\ell=0$.
In order to do this, we first need to characterize the eigenvalue $\ell=0$ for the problem (\ref{eigenvalueproblem}) arising from the linearization of NLS about the $T$-periodic solution  (\ref{U0def}). 

\subsection{The case $\epsilon=0$}
\label{eps=0}

 In this case, we show that the algebraic multiplicity of the zero eigenvalue is four by showing that the lowest order nonzero derivative of the Evans function (\ref{EvansDef}) at $\ell=0$ is the fourth derivative. More precisely, we have the following theorem

\begin{The}
\label{the1}
The Evans function $M_0(\ell)$ defined in (\ref{EvansDef})  is such that
$$
M_0(0)=M_0'(0)=M_0''(0)=M_0'''(0),
$$
where the primes denote derivatives with respect to $\ell$.
Furthermore, the fourth derivative, which is given by
\begin{equation}
\label{M04r}
 {M}_0^{(4)}(0)=384\,{\frac { \left( E^2 + K(K-2E)k'^2 \right) ^{2}}{
  k^4\,k'^4\,{\delta}^{8}
}},
\end{equation}
is nonzero for all $k\in (0,1)$.
Thus, the eigenvalue $\ell=0$ is of algebraic multiplicity exactly four.
\end{The}

%{\noindent}
%Note that the above quantity is plotted as a function of $k$ in Figure \ref{fig1} and we see that  $\widetilde{M}_0^{(4)}(0)>0$,
%

 \begin{proof}
 Let
 ${\rm\mathbf{w}}_i(x;\ell),\;i=1\ldots 4$,  be four linearly independent solutions of (\ref{fourbysys0}) and let
 $W(x;\ell)$ be a fundamental matrix solution of (\ref{fourbysys0}) with columns ${\rm\mathbf{w}}_i$.
 In order to compute the derivative of the Evans function at $\ell=0$, we expand each of the solutions as
 \begin{equation}
 {\rm\mathbf{w}}_i(x;\ell)={\rm\mathbf{w}}_i^0(x)+\ell {\rm\mathbf{w}}_i^1(x)+\ell^2 {\rm\mathbf{w}}_i^2(x)+{\cal{O}}(\ell^3).
 \label{uexp}
 \end{equation}
 The ${\rm\mathbf{w}}_i^0(x)$ are solutions of (\ref{fourbysys0}) for $\ell=0$, so
they are  linear combinations of the ${\rm\mathbf{V}}_i^h$, $i=1\ldots 4$ given in (\ref{A0sol}) and in Appendix \ref{app:B}. We make the choice
 $$
 {\rm\mathbf{w}}_i^0(x)= {\rm\mathbf{V}}_i^h(x).
 $$

Recursively, the higher order terms  ${\rm\mathbf{w}}_i^j$ for $j \ge1$ satisfy the non-homogeneous system
  \begin{equation}
 \label{NHS2}
 \dfrac{\rd}{\rd x} {\rm\mathbf{w}}_i^j={\cal A}(0,x)\,{\rm\mathbf{w}}_i^j+{\bf{M}}_i^{j-1},
\end{equation}
where
$$
{\bf{M}}_i^j=
 \begin{bmatrix}
 0\\0\\
  -\ri\,\left({\rm\mathbf{w}}_i^j\right)_1\\
  \ri\,\left({\rm\mathbf{w}}_i^j\right)_2
 \end{bmatrix},
$$
 and $\left({\rm\mathbf{w}}_i^j\right)_k$ denotes the $k$th component of the vector ${\rm\mathbf{w}}_i^j$.
 Since we have a fundamental set of solutions for the homogeneous system, we can solve \eqref{NHS2} by variation of parameters method. For the higher order terms ${\rm\mathbf{w}}_i^j$ we choose
 $$
 {\rm\mathbf{w}}^{j}_i(x)=W_0(x)\,\int_0^xW_0^{-1}(z)\,{\bf{M}}_i^j(z)\rd z,\;\;\;i=1\ldots4,\;\;j\ge 1,
 $$
 where $W_0(x)\equiv W(x;0)$ is the fundamental matrix solution for the system (\ref{fourbysys0}) at $\ell=0$ with columns ${\rm\mathbf{V}}_i^h$. The formulas for $ {\rm\mathbf{w}}^{1}_1$ and $ {\rm\mathbf{w}}^{1}_2$ can be written in a compact way as
\begin{equation}
\label{w11w12}
{\rm\mathbf{w}}^1_1=
\frac{\ri}{2}\left(
 \begin{bmatrix}
 -x\,U_0\\\,x\,U_0\\
 -\,\left(x\,U_0\right)'\\
  \,\left(x\,U_0\right)'
 \end{bmatrix}+{\rm\mathbf{V}}_4^{h}\right),\;\;\;
 {\rm\mathbf{w}}^1_2=
\frac{\ri}{2\delta^2(1-2k^2)}
  \left(\begin{bmatrix}
 \left(x\,U_0\right)'\\
 \left(x\,U_0\right)'\\
 \left(x\,U_0\right)''\\
 \left(x\,U_0\right)''
 \end{bmatrix} +{\rm\mathbf{V}}_3^{h}
 \right),
 \end{equation}
 where $U_0$ is the solution of NLS  given in (\ref{U0def}).  The expressions
 for ${\rm\mathbf{w}}_3^1$, ${\rm\mathbf{w}}_4^1$, ${\rm\mathbf{w}}_1^2$, and ${\rm\mathbf{w}}_2^2$ are given
 in Appendix \ref{app:B}.

 In order to compute the Evans function (\ref{EvansDef}), we need  solutions
 ${\mathbf{\vphi}}_i(x;\ell),\;i=1\ldots4$  of (\ref{fourbysys0}) with initial condition ${\mathbf{\vphi}}_i(0,\ell)={\bf{e}}_i$, where ${\bf{e}}_i$ are the standard basis vectors of $\mathbb{R}^4$. The
 ${\mathbf{\vphi}}_i(x;\ell)$ are thus the columns of the matrix $\Phi(x;\ell)$ used in the definition of the Evans function (\ref{EvansDef1}). Matching
 the latter initial conditions with those of the ${\rm\mathbf{w}}_i$, one finds that
  \begin{eqnarray*}
 {{\vphi}}_1&=&\frac{1}{2\delta k}\,\left( {\rm\mathbf{w}}_3+{\rm\mathbf{w}}_2 \right),\\
 {\mathbf{\vphi}}_2&=&\frac{1}{2\delta k}\,\left( {\rm\mathbf{w}}_3-{\rm\mathbf{w}}_2 \right),\\
 {\mathbf{\vphi}}_3&=&\frac{1}{2\delta^3k}\,\left(\delta^2 {\rm\mathbf{w}}_4-{\rm\mathbf{w}}_1 \right),\\
 {\mathbf{\vphi}}_4&=&\frac{-1}{2\delta^3k}\,\left(\delta^2 {\rm\mathbf{w}}_4+{\rm\mathbf{w}}_1 \right).
 \end{eqnarray*}
 The Evans function (\ref{EvansDef}) can then be expressed as
 \begin{equation}
 \label{rel}
 M_0(\ell)=\frac{1}{4\delta^6k^4}\widetilde{M}_0(\ell)
 \end{equation}
 where
 \begin{equation}
 \widetilde{M}_0(\ell)={\mbox{det}}\left(W(T;\ell)-A\right),\;\;A=
\begin{bmatrix}
0&\delta k&\delta k&0\\
0                &-\delta  k    &\delta k &0\\
-\delta^3 k&0                 &0           &\delta k\\
\delta^3k &0                 & 0          &\delta k
\end{bmatrix}
 \end{equation}
 and $W(x;\ell)$ is the fundamental matrix solution of (\ref{fourbysys0}) with columns ${\rm\mathbf{w}}_i,\;i=1\ldots4$. Note that $W(0;\ell)=A$.

 Since  $\widetilde{M}_0(\ell)$ differs from ${M}_0(\ell)$ by a factor which is constant in $\ell$, the derivatives of $M_0$ at $\ell=0$ are constant multiples of those of  $\widetilde{M}_0$. The derivatives of  $\widetilde{M}_0$ at $\ell=0$ can be computed using the expansion  (\ref{uexp}).
 Because  ${\rm\mathbf{V}}^h_i, i=1,2$ are $T$-periodic, the first two columns of $W(T;\ell)-A$ are zero at $\ell=0$. Thus, the expansion of the determinant $\widetilde{M}_0$ in powers of $\ell$ has no terms of lower order than $\ell^2$. Furthermore, the coefficient of $\ell^2$ is given by
 \begin{equation}
 \label{M0pp}
 \frac{1}{2}\widetilde{M}_0''(0)=
{\mbox{det}}\left({\rm\mathbf{w}}^1_1(T),{\rm\mathbf{w}}^1_2(T),{\rm\mathbf{V}}^h_3(T)-A_3,{\rm\mathbf{V}}^h_4(T)-A_4\right),
 \end{equation}
 where $A_i$ denotes the $i$th column of $A$.

 The right-hand-side of (\ref{M0pp}) is easily shown to be zero, as follows.
 Using the expressions for ${\rm\mathbf{w}}_{1}^1$ and ${\rm\mathbf{V}}^{h}_4$ given in (\ref{w11w12}) and in Appendix \ref{app:B} we find that the vector
 \begin{equation}
\label{geneig1}
{\rm\mathbf{w}}_{g,1}\equiv{\rm\mathbf{w}}_{1}^1+\frac{\ri}{2}E\left(k'^2K-E\right){\rm\mathbf{V}}^{h}_4
 \end{equation}
is a $T$-periodic function of $x$. Taking the difference of the values of the vector ${\rm\mathbf{w}}_{g,1}$ at $x=0$ and $x=T$ gives a linear combination of the first and fourth columns of the matrix in (\ref{M0pp}) which must be zero.  Similarly, the vector
 \begin{equation}
\label{geneig2}
{\rm\mathbf{w}}_{g,2}\equiv{\rm\mathbf{w}}_{2}^1-\frac{\ri}{2\delta^2}E\left(k'^2K+E(1-2k^2)\right){\rm\mathbf{V}}^{h}_3,
 \end{equation}
 is periodic and thus a linear combination of the second and third columns of the matrix in (\ref{M0pp}) is zero.

 These facts also imply that $\widetilde{M}_0'''(0)=0$, as follows. The idea is to write $\widetilde{M}_0(\ell)$ as
 \begin{equation}
 \label{M0exp}
 \frac{\widetilde{M}_0(\ell)}{\ell^2}={\mbox{det}}\left({\rm\mathbf{w}}^1_1(T)+\mathcal{O}(\ell),\;
 {\rm\mathbf{w}}^1_2(T)+\mathcal{O}(\ell),\;
 {\rm\mathbf{V}}^h_3(T)-A_3+\mathcal{O}(\ell),\;
 {\rm\mathbf{V}}^h_4(T)-A_4+\mathcal{O}(\ell)\right).
 \end{equation}
 Since the $\ell^1$ term of the right-hand-side of (\ref{M0exp}) will be a sum of determinants of matrices each involving three of the four columns of the matrix in (\ref{M0pp}), it must be zero. Thus, we have that the Evans function for $\epsilon=0$ at least has a fourth order zero at $\ell=0$.

 Then,
  \begin{equation}
  \begin{aligned}
 \label{M0expc}
 \frac{\widetilde{M}_0(\ell)}{\ell^2}&=&{\mbox{det}}\left[\right.&\left.{\rm\mathbf{w}}^1_1(T)+{\rm\mathbf{w}}^2_1(T)\ell+\mathcal{O}(\ell^2),\;
 {\rm\mathbf{w}}^1_2(T)+{\rm\mathbf{w}}^2_2(T)\ell+\mathcal{O}(\ell^2),\right.\\
&&&\left.{\rm\mathbf{V}}^h_3(T)-A_3+{\rm\mathbf{w}}^2_3(T)\ell+\mathcal{O}(\ell^2),\;
 {\rm\mathbf{V}}^h_4(T)-A_4+{\rm\mathbf{w}}^2_4(T)\ell+\mathcal{O}(\ell^2)\right]\\
 &=&\ell^2{\mbox{det}}\left(\right.&\left.{\rm\mathbf{w}}^2_1(T),{\rm\mathbf{w}}^2_2(T),{\rm\mathbf{V}}^h_3(T)-A_3,{\rm\mathbf{V}}^h_4(T)-A_4\right)\\
 &+&\ell^2{\mbox{det}}\left(\right.&\left.{\rm\mathbf{w}}^2_1(T),{\rm\mathbf{w}}^1_2(T),{\rm\mathbf{w}}^1_3(T),{\rm\mathbf{V}}^h_4(T)-A_4\right)\\
  &+&\ell^2{\mbox{det}}\left(\right.&\left.{\rm\mathbf{w}}^1_1(T),{\rm\mathbf{w}}^2_2(T),{\rm\mathbf{V}}^h_3(T)-A_3,{\rm\mathbf{w}}^1_4(T)\right)\\
   &+&\ell^2{\mbox{det}}\left(\right.&\left.{\rm\mathbf{w}}^1_1(T),{\rm\mathbf{w}}^1_2(T),{\rm\mathbf{w}}^1_3(T),{\rm\mathbf{w}}^1_4(T)\right) +\mathcal{O}(\ell^3).
 \end{aligned}
 \end{equation}
 It is then straightforward to use (\ref{M0expc}), and the expressions for ${\rm\mathbf{w}}^1_i$, ${\rm\mathbf{w}}^2_i$, and ${\rm\mathbf{V}}^h_i$ in (\ref{w11w12}) and in Appendix \ref{app:B}, to find
 \begin{equation}
\label{M04}
 \widetilde{M}_0^{(4)}(0)=1536\,{\frac { \left( E^2 + K(K-2E)k'^2 \right) ^{2}}{
  {\delta}^{2} k'^4}}.
\end{equation}
 Then one obtains  (\ref{M04r}) using the relation (\ref{rel}).  Using the inequalities
$K > E > k' K$, which hold for all moduli $k \in (0,1)$, we find that
$$E^2 + K(K-2E)k'^2 > 2 k'^2 K (K-E) > 0,$$
and thus the fourth derivative of the Evans function is nonzero for all such $k$.
\end{proof}

%
%\begin{figure}[h]
%\scalebox{0.5}{\includegraphics{fig4.eps}}
%\caption{The graph of $M_0^{(4)}(0)$ given in (\ref{M04}) as a function of the elliptic modulus $k$ in the case $\delta=1$.
%\label{fig1}}
%\end{figure}

We have established in \S \ref{PS1} that for for $\ell=0$ (\ref{fourbysys0}) only has a two-dimensional space of $T$-periodic solutions generated by ${\rm\mathbf{V}}_1^h$ and ${\rm\mathbf{V}}_2^h$. The eigenvalue problem
  (\ref{eigenvalueproblem}) thus admits a two-dimensional eigenvector space for $\ell=0$. However, since the zero of the Evans function is of multiplicity four, the eigenvalue $\ell=0$ actually is of algebraic multiplicity four. Below we describe two eigenvectors and two generalized eigenvectors corresponding to $\ell=0$.

  We denote the two eigenvectors $
{\bf{\vomega}}_0^1,\;{\bf{\vomega}}_0^2\in {\mbox{Ker}}(\mathcal{L}_0)
$. They are given by
\begin{equation}
\label{psi012def}
{\bf{\vomega}}_0^1= \begin{bmatrix}
 U_0'\\U_0'
 \end{bmatrix},\;\;\;\;
{\bf{\vomega}}_0^2=
 \begin{bmatrix}
 U_0\\-U_0
 \end{bmatrix},
  \end{equation}
where $U_0$ is the NLS solution (\ref{U0def}).
Furthermore, there are two generalized eigenvectors
$
{\bf{\vomega}}^{g,1}_0,\;{\bf{\vomega}}^{g,2}_0$ satisfying
\begin{equation}
\label{genpro}
\mathcal{L}_0\,{\bf{\vomega}}^{g,i}_0={\bf{\vomega}}_0^{i},\;\;\;i=1,2.
\end{equation}
The two components of the generalized eigenvector ${\bf{\vomega}}^{g,1}_0$ (resp. ${\bf{\vomega}}^{g,2}_0$) are the first two components of ${\rm\mathbf{w}}_{g,1}$ (resp.  ${\rm\mathbf{w}}_{g,2}$) given in (\ref{geneig1}) (resp. (\ref{geneig2})).

\subsection{The case $\epsilon\neq0$}
When $\epsilon$ is not zero, the algebraic multiplicity of the eigenvalue $\ell=0$ is only two and two eigenvalues move away from the origin. In order to track them, we use an argument based on the Fredholm alternative.

The linearization of CGL about the $T$-periodic  solution (\ref{Udef}) gives rise to the eigenvalue problem (\ref{CGLeigenvalueproblem})
defined over the space $\left(L_2(T)\right)^2$ of $T$-periodic square integrable vector valued functions.
The linear operator ${\mathcal{L}}$ has two eigenvectors
${\bf{\vomega}}^1,\;{\bf{\vomega}}^2\in {\mbox{Ker}}(\mathcal{L})$. They are given by
\begin{equation}
\label{CGLev1}
{\bf{\vomega}}^1= \begin{bmatrix}
 U'\\U'
 \end{bmatrix},\;\;
{\bf{\vomega}}^2=
 \begin{bmatrix}
 U\\-U
 \end{bmatrix},
 \end{equation}
 where $U$ is the solution of CGL given in (\ref{Udef}).
 The eigenvectors (\ref{CGLev1}) can be expanded in $\epsilon$ as
 \begin{equation}
\begin{aligned}
\label{CGLev}
{\bf{\vomega}}^1=
{\bf{\vomega}}_0^1+\epsilon{\bf{\vomega}}_1^1+\epsilon^2{\bf{\vomega}}_2^1+\mathcal{O}(\epsilon^3),\\
{\bf{\vomega}}^2=
{\bf{\vomega}}_0^2+\epsilon{\bf{\vomega}}_1^2+\epsilon^2{\bf{\vomega}}_2^2+\mathcal{O}(\epsilon^3),
 \end{aligned}
 \end{equation}
where ${\bf{\vomega}}_0^i$ are the two vector-valued functions of ${\mbox{Ker}}(\mathcal{L}_0)$ given in (\ref{psi012def}), and ${\bf{\vomega}}_1^i$ are given by
\begin{equation}
\label{psi112def}
{\bf{\vomega}}_1^1= \begin{bmatrix}
 U_1'\\U_1'
 \end{bmatrix},\;\;\;\;
{\bf{\vomega}}_1^2=
 \begin{bmatrix}
 U_1\\-U_1
 \end{bmatrix},
  \end{equation}
 where $U_1$ is the coefficient of $\epsilon$ in the solution of CGL and is given in  (\ref{U1def}). Furthermore, the vector-valued functions ${\bf{\vomega}}_0^i$, ${\bf{\vomega}}_1^i$, and ${\bf{\vomega}}_2^i$ in (\ref{CGLev}) satisfy the differential equations
 \begin{equation}
 \label{nicerel}
\begin{aligned}
&\mathcal{L}_0{\bf{\vomega}}_1^i+\mathcal{L}_1{\bf{\vomega}}_0^i=0,\\
&\mathcal{L}_1{\bf{\vomega}}_1^i+\mathcal{L}_2{\bf{\vomega}}_0^i+\mathcal{L}_0{\bf{\vomega}}_2^i=0,\;\;i=1,2.
\end{aligned}
\end{equation}
These equations are obtained by inserting the right-hand-sides of (\ref{CGLev}) into (\ref{CGLeigenvalueproblem}) with $\ell=0$ and then writing the relations occurring at the first and second order in $\epsilon$.

In order to track the eigenvalues of $\mathcal{L}$ emerging from the origin, we suppose that ${\bf{\vomega}}_{\epsilon}$ is an eigenvector corresponding the a nonzero eigenvalue $\ell_{\epsilon}$. We expand these in $\epsilon$:
\begin{equation}
\label{eigLe}
\begin{aligned}
&\ell_{\epsilon}=\epsilon \ell_1+\epsilon^2\ell_2+\mathcal{O}(\epsilon^3),\\
&{\bf{\vomega}}_{\epsilon}={\bf{\vomega}}_{0}+\epsilon{\bf{\vomega}}_{1}+\epsilon^2{\bf{\vomega}}_{2}+\mathcal{O}(\epsilon^3).
\end{aligned}
\end{equation}
The vector ${\bf{\vomega}}_{0}$ is in the kernel of $\mathcal{L}_0$. Thus ${\bf{\vomega}}_{0}$ will be some linear combination of ${\bf{\vomega}}^{1}_0$ and ${\bf{\vomega}}^{2}_0$:
\begin{equation}
\label{lincomb}
{\bf{\vomega}}_{0}=\alpha{\bf{\vomega}}^{1}_0+\beta{\bf{\vomega}}^{2}_0,
\end{equation}
where $\alpha$ and $\beta$ are complex constants. We insert the expressions (\ref{eigLe}) into the
eigenvalue problem (\ref{CGLeigenvalueproblem}) for $\mathcal{L}$. At first order in $\epsilon$, ${\bf{\vomega}}_{1}$ satisfies the following equation:
\begin{equation}
\mathcal{L}_0{\bf{\vomega}}_{1}=\ell_1{\bf{\vomega}}_{0}-\mathcal{L}_1{\bf{\vomega}}_{0}.
\label{psi1}
\end{equation}
Using (\ref{genpro}),  the first equation in (\ref{nicerel})
and linear superposition, we obtain the following solution to  (\ref{psi1})
\begin{equation}
\label{om1}
{\bf{\vomega}}_{1}=\alpha{\bf{\vomega}}_{1}^1+\beta{\bf{\vomega}}_1^{2}+\ell_1\left(\alpha{\bf{\vomega}}^{g,1}_0+\beta{\bf{\vomega}}^{g,2}_0\right),
\end{equation}
where $\alpha$ and $\beta$ are the constants in (\ref{lincomb}), ${\bf{\vomega}}_{1}^i$ are given in (\ref{psi112def}), and ${\bf{\vomega}}^{g,i}_0$ are the generalized eigenvectors of $\mathcal{L}_0$ satisfying the relations (\ref{genpro}). At second order, one finds  the following equation for ${\bf{\vomega}}_{2}$:
\begin{equation}
\label{psi2}
\mathcal{L}_0{\bf{\vomega}}_{2}=\ell_1{\bf{\vomega}}_{1}+\ell_2{\bf{\vomega}}_{0}-\mathcal{L}_1{\bf{\vomega}}_{1}
-\mathcal{L}_2{\bf{\vomega}}_{0}.
\end{equation}
We can eliminate the $\mathcal{L}_2$ term from (\ref{psi2}). Indeed, if we use the second part of (\ref{nicerel}) and the expression for ${\bf{\vomega}}_{1}$ given in (\ref{om1}), one finds that
\begin{equation}
\mathcal{L}_1{\bf{\vomega}}_{1}
+\mathcal{L}_2{\bf{\vomega}}_{0}=
\ell_1\mathcal{L}_1 \left(\alpha{\bf{\vomega}}^{g,1}_0+\beta{\bf{\vomega}}^{g,2}_0\right)
-\mathcal{L}_0 \left(\alpha{\bf{\vomega}}^{1}_2+\beta{\bf{\vomega}}^{2}_2\right).
\end{equation}
Equation (\ref{psi2}) then becomes
\begin{equation}
\label{psi22}
\mathcal{L}_0{\bf{\vomega}}_{2}=\ell_1{\bf{\vomega}}_{1}+\ell_2{\bf{\vomega}}_{0}-\ell_1\mathcal{L}_1 \left(\alpha{\bf{\vomega}}^{g,1}_0+\beta{\bf{\vomega}}^{g,2}_0\right)+
\mathcal{L}_0 \left(\alpha{\bf{\vomega}}^{1}_2+\beta{\bf{\vomega}}^{2}_2\right).
\end{equation}

We will not solve (\ref{psi22}) but rather find a solvability condition that guarantees $T$-periodic solutions. First, we define the usual inner product
\begin{equation}
\label{scalpro}
\left<{\bf{X}}_1,{\bf{X}}_2\right>\equiv \int_0^T {\bf{X}}_1^T(x) \overline{{\bf{X}}}_2(x) \,\rd x,\;\;{\bf{X}}_i\in \left(L_2(T)\right)^2.
\end{equation}
With respect to (\ref{scalpro}), the adjoint of $\mathcal{L}_0$ from (\ref{Ldef}) is given by
\begin{equation}
\label{Ldagdef}
{\mathcal L}^{\dagger}=-\left[\frac{\rd^2}{\rd x^2}+ \begin{pmatrix}4U_0^2+\alpha & 2 U_0^2 \\
2 U_0^2 & 4U_0^2+\alpha
\end{pmatrix}\right]J,\;\;\;J=\begin{pmatrix} \ri & 0 \\ 0 & -\ri \end{pmatrix}.
\end{equation}
It is then easy to verify that the kernel of  $\mathcal{L}_0^{\dagger}$ is generated by $J{\bf{\vomega}}^{i}_0,i=1,2$, where ${\bf{\vomega}}^{i}_0$ are the generators of the kernel of  $\mathcal{L}_0$.

We can now obtain compatibility conditions for (\ref{psi22}) by taking the scalar product with elements of the kernel of  $\mathcal{L}_0^{\dagger}$:
\begin{equation}
\label{dotpsi2}
\begin{aligned}
\left<\mathcal{L}_0{\bf{\vomega}}_{2},\;J{\bf{\vomega}}^{i}_0\right>&=
\ell_1\left<{\bf{\vomega}}_{1},\;J{\bf{\vomega}}^{i}_0\right>+\ell_2\left<{\bf{\vomega}}_{0},\;J{\bf{\vomega}}^{i}_0\right>-
\ell_1\left<\mathcal{L}_1 \left(\alpha{\bf{\vomega}}^{g,1}_0+\beta{\bf{\vomega}}^{g,2}_0\right),\;J{\bf{\vomega}}^{i}_0\right>\\
&+\left<\mathcal{L}_0 \left(\alpha{\bf{\vomega}}^{1}_2+\beta{\bf{\vomega}}^{2}_2\right),\;J{\bf{\vomega}}^{i}_0\right>,\;\;i=1,2.
\end{aligned}
\end{equation}
Because $J{\bf{\vomega}}^{i}$ is in the kernel of $\mathcal{L}_0^{\dagger}$, the left-hand side and the last term of the right-hand side of (\ref{dotpsi2}) are zero. We can then rewrite (\ref{dotpsi2}) using the expression for ${\bf{\vomega}}_{1}$ given in (\ref{om1})
\begin{equation}
\label{dotpsi3}
\begin{aligned}
0&=
\ell_1^2
\left<\alpha{\bf{\vomega}}^{g,1}_0+\beta{\bf{\vomega}}^{g,2}_0,\;J{\bf{\vomega}}^{i}_0\right>+
\ell_1\left(\left<\alpha{\bf{\vomega}}_{1}^1+\beta{\bf{\vomega}}_1^{2},\;J{\bf{\vomega}}^{i}_0\right>-
\left<\mathcal{L}_1 \left(\alpha{\bf{\vomega}}^{g,1}_0+\beta{\bf{\vomega}}^{g,2}_0\right),\;J{\bf{\vomega}}^{i}_0\right>\right)\\
&+
\ell_2\left<\alpha{\bf{\vomega}}^{1}_0+\beta{\bf{\vomega}}^{2}_0,\;J{\bf{\vomega}}^{i}_0\right>,\;\;i=1,2.
\end{aligned}
\end{equation}

This is a homogeneous linear system of equations for $\alpha$ and $\beta$ and we are looking for the values of $\ell_i$ for which \eqref{dotpsi3} has a non-trivial solution.
The system can be simplified further by making the observation that the components of the vectors  ${\bf{\vomega}}^{i}_0$, ${\bf{\vomega}}^{i}_1$, and ${\bf{\vomega}}^{g,1}_0$ are even functions of $x$ with respect to the axis $x=T/2$ for $i=1$ and odd for $i=2$.  Thus, any inner product in (\ref{dotpsi3}) involving two different superscripts will be zero. Furthermore, because $J$ is skew-adjoint, the last term of (\ref{dotpsi3}) is always zero.  We reduce to the two equations
\begin{equation}
\label{dotpsi4}
\begin{aligned}
0&=
\alpha\ell_1\left(\ell_1
\left<{\bf{\vomega}}^{g,1}_0,\;J{\bf{\vomega}}^{1}_0\right>+
\left<{\bf{\vomega}}_{1}^1,\;J{\bf{\vomega}}^{1}_0\right>-
\left<\mathcal{L}_1 {\bf{\vomega}}^{g,1}_0,\;J{\bf{\vomega}}^{1}_0\right>\right),\\
0&=
\beta\ell_1\left(\ell_1
\left<{\bf{\vomega}}^{g,2}_0,\;J{\bf{\vomega}}^{2}_0\right>+
\left<{\bf{\vomega}}_{1}^2,\;J{\bf{\vomega}}^{2}_0\right>-
\left<\mathcal{L}_1 {\bf{\vomega}}^{g,2}_0,\;J{\bf{\vomega}}^{2}_0\right>\right).\\
\end{aligned}
\end{equation}
The system (\ref{dotpsi4}) has a non-trivial solution $(\alpha,\beta)$ if $\ell_1=0$, which corresponds to the generators of the kernel of $\mathcal{L}$ given in (\ref{CGLev1}), or if $\ell_1$ is given by either of the expressions
\begin{equation}
\label{ell1}
\begin{aligned}
\ell_1=\frac{\left<\mathcal{L}_1 {\bf{\vomega}}^{g,1}_0,\;J{\bf{\vomega}}^{1}_0\right>-\left<{\bf{\vomega}}_{1}^1,\;J{\bf{\vomega}}^{1}_0\right>
}{\left<{\bf{\vomega}}^{g,1}_0,\;J{\bf{\vomega}}^{1}_0\right>}\;\;{\mbox{or}}\;\;
\ell_1=\frac{\left<\mathcal{L}_1 {\bf{\vomega}}^{g,2}_0,\;J{\bf{\vomega}}^{2}_0\right>-\left<{\bf{\vomega}}_{2}^1,\;J{\bf{\vomega}}^{2}_0\right>
}{\left<{\bf{\vomega}}^{g,2}_0,\;J{\bf{\vomega}}^{2}_0\right>}.\end{aligned}
\end{equation}

The components of the vectors ${\bf{\vomega}}^{i}_0$ and ${\bf{\vomega}}^{i}_1$ are real and  those of ${\bf{\vomega}}^{g,i}_0$ are purely imaginary making the expressions in (\ref{ell1}) all real. Thus (\ref{CGLeigenvalueproblem}) has an eigenvalue on the right side of the complex plane for small positive values of $\epsilon$ if one of the two solutions in (\ref{ell1}) is positive.  This gives an instability condition for the solution (\ref{Udef}).

Note that in what follows, we make the simplifying assumption that $\delta=1$. This is done without loss of generality because CGL in the form (\ref{CGL}) has the symmetry that if $q_2(x,t)=\overline{q_1(x,t)}=f(x,t)$ is a solution for a given value of the coefficient $r=r_0$, then $q_2(x,t)=\overline{q_1(x,t)}=\gamma f(\gamma\,x,\gamma^2 t)$ is a solution corresponding to the value of the coefficient $r=r_0/\gamma^2$. This symmetry can thus be used to set $\delta$ to 1.

The condition that the first expression in (\ref{ell1}) be positive gives rise to the inequality
\begin{equation}
\label{condi1}
\begin{aligned}
%&-900\,{{\it K}}^{4}-420\,{{\it E}}^{4}+1712\,s\,k^6{\it E}\,{{\it K}}^{3}+
%2688\,s\,k^2(k^2-1){{\it E}}^{2}\,{{\it K}}^{2}\\
%&+3424\,sk^2{\it E}\,{{\it K}}^{3}
%-5256\,s\,{k}^{4}{\it E}\,{{\it K}
%}^{3}+1748\,s{k}^{4}{{\it K}}^{4}-832\,s{k}^{2}{{\it K}}^{4}
%-856\,s\,{k}^{6}{{\it K}}^{4}\\
%&+1680\,(1-k^2){{\it E}}^{3}{\it K}+2832
%\,{\it E}\,{{\it K}}^{3}-420\,s{{\it E}}^{4}+32\,{k}^{6}{\it E}\,{{\it K}}^{3}
%-672\,{k}^{4}{{\it E}}^{2}{{\it K}}
%^{2}
%\\&+624\,{k}^{4}{\it E}\,{{\it K}}^{3}-3296\,{k}^{2}{\it E}\,{{\it K}}
%^{3}+3192\,({k}^{2}-1){{\it E}}^{2}{{\it K}}^{2}
%+840\,s{{\it E}}^{3}{
%\it K}
%\\&-672\,s{{\it E}}^{2}{{\it K}}^{2}+312\,s{\it E}\,{{
%\it K}}^{3}
%-16\,{k}^{6}{{\it K}}^{4}
%-772\,{k}^{4}{{\it K}}^{4}+1688\,{k}^{2}{{\it K}
%}^{4}-60\,s{{\it K}}^{4}\\
&4k'^2 ((214k^4 -223 k^2-15)s+4k^4+197 k^2-225 )K^4-420(1+s)E^4\\
&+8((214 k^6-657 k^4 +428 k^2+39 )s+4 k^6+78 k^4-412 k^2+354) E K^3\\
&-168 (4 k^4+16 s k^2 k'^2+4s+19k'^2) E^2 K^2+840(s+2k'^2) E^3 K
\\
&-210\, \left( 1+s \right){\it K}\, \left( (1+k^2){{\it E}}
^{2}+(1+7k^2+k^4){{\it K}}^{2}
-2\,(1+4\,k^2+k^4)E\,K
\right)  {\it I_1}
\\
&-315\,{k}^{4} \left( 1+s \right){{\it K}}^{2}
 \left( {\it E}-5\,{\it K}-6\,{k}^{2}{\it K} \right) {\it I_2
} -1575\,{k}^{6} \left( 1+s \right){{\it K}}^{3}
 {\it I_3}> 0.
 \end{aligned}
\end{equation}
where
$$
I_1\equiv\int_{0}^{4K}Z^2(x)\rd x,\;\;I_2\equiv\int_{0}^{4K}Z^2(x)\sn^4(x)\rd x,\;\;I_3\equiv\int_{0}^{4K}Z^2(x)\sn^6(x)\rd x.
$$
Furthermore, we obtain the following inequality by requiring the second expression to be positive
\begin{equation}
\label{condi2}
\begin{aligned}
%&-324\,{k}^{4}{{\it K}}^{4}-676\,{{\it K}}^{4}-128\,{k}^{6}{{\it K}}
%^{4}+1128\,{k}^{2}{{\it K}}^{4}-3196\,s\,{{\it K}}^{4}-7884\,s\,{k}^{4}{{\it K}
%}^{4}
%\\&
%-20992\,s\,{k}^{2}{\it E}\,{{\it K}}^{3}+8688\,s\,{k}^{2}{{\it K}}^{4}+
%2392\,s\,{k}^{6}{{\it K}}^{4}+392\,{k}^{2}(k^2-1)E^2{{\it K}}^{2}
%\\&
%+1680\,s\,{{\it E}}^{3}{\it K}
%+17248\,s\,{k}^{2}{{\it E}}^{2}{{\it K}
%}^{2}+1544\,E\,{{\it K}}^{3}+256\,{k}^{6}E\,{{\it K}}^{3}-840\,{{\it E}}^{3}{\it K}
%\\&
%-448\,{{\it E}}^{2}{{\it K}}^{2}
%-8008\,s{{\it E}}^{2}{{\it K}}^{2}+9104\,s\,E\,{{\it K
%}}^{3}-4784\,s\,{k}^{6}{\it E}\,{{\it K}}^{3}
%+16864\,s\,{k}^{4}{\it E}\,{{\it K}}^{3}
%\\&
%-9688\,s\,{k}^{4
%}{{\it E}}^{2}{{\it K}}^{2}-5040\,s\,{k}^{2}{{\it E}}^{3}{\it K}-
%776\,{k}^{4}{\it E}\,{{\it K}}^{3}-832\,{k}^{2}{\it E}\,{{\it K}}^{3}+
%420\,(1+s){{\it E}}^{4}
&4k'^2 ((598 k^4-1373  k^2+799)s -32 k^4-113 k^2+169) K^4-420(1+s) E^4\\
&8 ((598   k^6-2108 k^4 +2624  k^2-1138) s-32 k^6+97 k^4+104 k^2-193) E K^3\\
&+56 ( (173 k^4-308 k^2+143)s+7k^2 k'^2+8)E^2 K^2
+840(6 s k^2-2 s+1) E^3 K
\\&
+210\, \left( 1+s \right){\it K}\, \left( (1-8\,{k}^{2}
-2\,{k}^{4}){{\it K}}^{2}+(4\,{k}^{4}+7\,{k}^{2}-2)E\,{\it K}+(1+{k}^{2}){{\it E}}^
{2} \right)  {\it I_1}
\\&
+315\,\left( 1+s \right){k}^{4}{{\it K}}^{2}   \left( 8\,{k}^{2}{\it
K}+{\it E}+3\,{\it K} \right) {\it I_2}
-1575\, \left( 1+s \right){k}^{6}{{\it K}}^{3} {\it
I_3}>0.
\end{aligned}
\end{equation}

\begin{The}\label{newthm}
When either condition (\ref{condi1}) or (\ref{condi2}) is met, one eigenvalue emerges from the origin on the right side of the complex $\ell$-space for small positive $\epsilon$. In other words, conditions (\ref{condi1}) and (\ref{condi2}) are spectral instability criteria for the solution of CGL given in (\ref{Udef}).
\end{The}
  
Figure \ref{fig2} shows the regions in the $k$-$s$ space in which the inequalities for existence (\ref{req}) and instability (\ref{condi1}) and (\ref{condi2}) are satisfied. Depending on the values of $k$ and $s$, there are two possible types of behaviors for the two emerging eigenvalues (see Figure \ref{fig7}). In the first case, the two eigenvalues move to the left on the real line making our analysis inconclusive since other eigenvalues could possibly emerge from zeroes of the Evans function that are elsewhere on the imaginary axis. In the second case,  the eigenvalues move on opposite directions on the real line causing the solution (\ref{Udef}) to be spectrally unstable.

\newlength{\fred}
\setlength{\fred}{.333\textwidth}
\begin{figure}
\centering
%\scalebox{0.5}{\includegraphics{fig1.eps}}
%\hspace{-5cm}(a)
%\hspace{3cm}
%\scalebox{0.5}{\includegraphics{fig2.eps}}
%\hspace{-7cm}(b)
%\hspace{3cm}
%\scalebox{0.5}{\includegraphics{fig3.eps}}\hspace{-4cm}(c)
%\hspace{4cm}
%\hspace{-.25in}
\begin{minipage}{\fred}
\includegraphics[width=\fred]{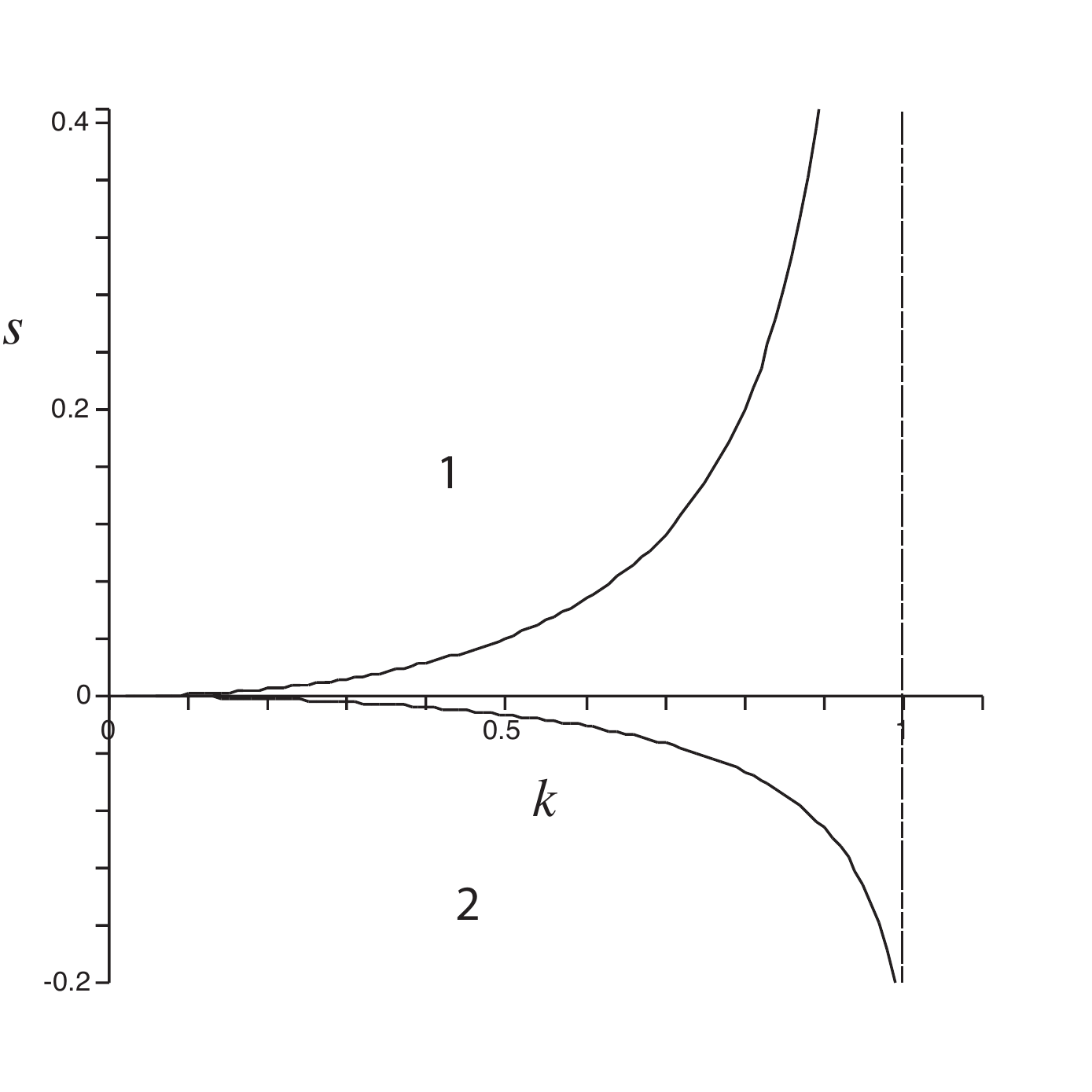}

\hfil (a) \hfill
\end{minipage}
\hspace{-.05\textwidth}
\begin{minipage}{\fred}
\includegraphics[width=\fred]{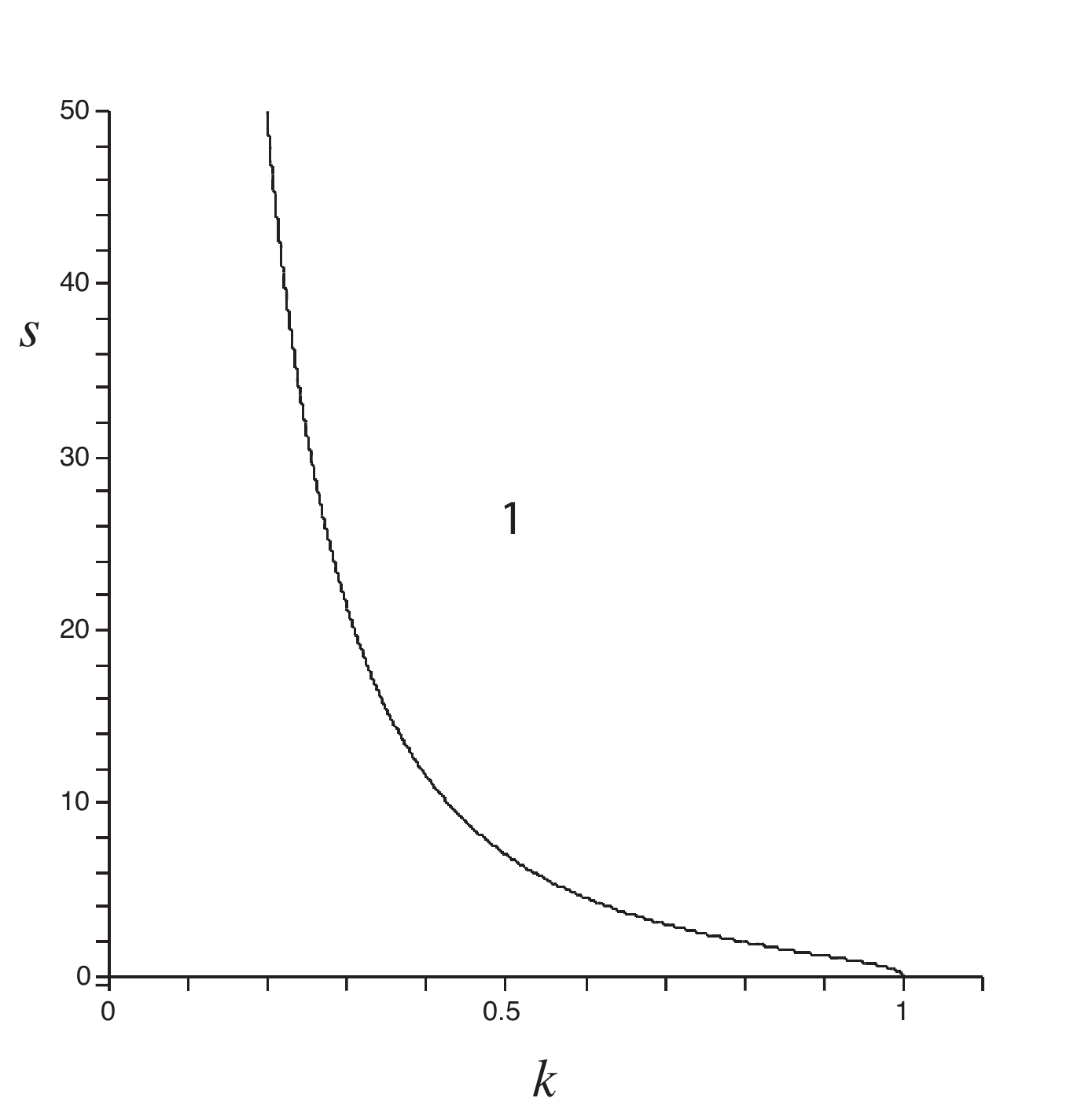}

\hfil (b) \hfill
\end{minipage}
\hspace{-.05\textwidth}
\begin{minipage}{\fred}
\includegraphics[width=\fred]{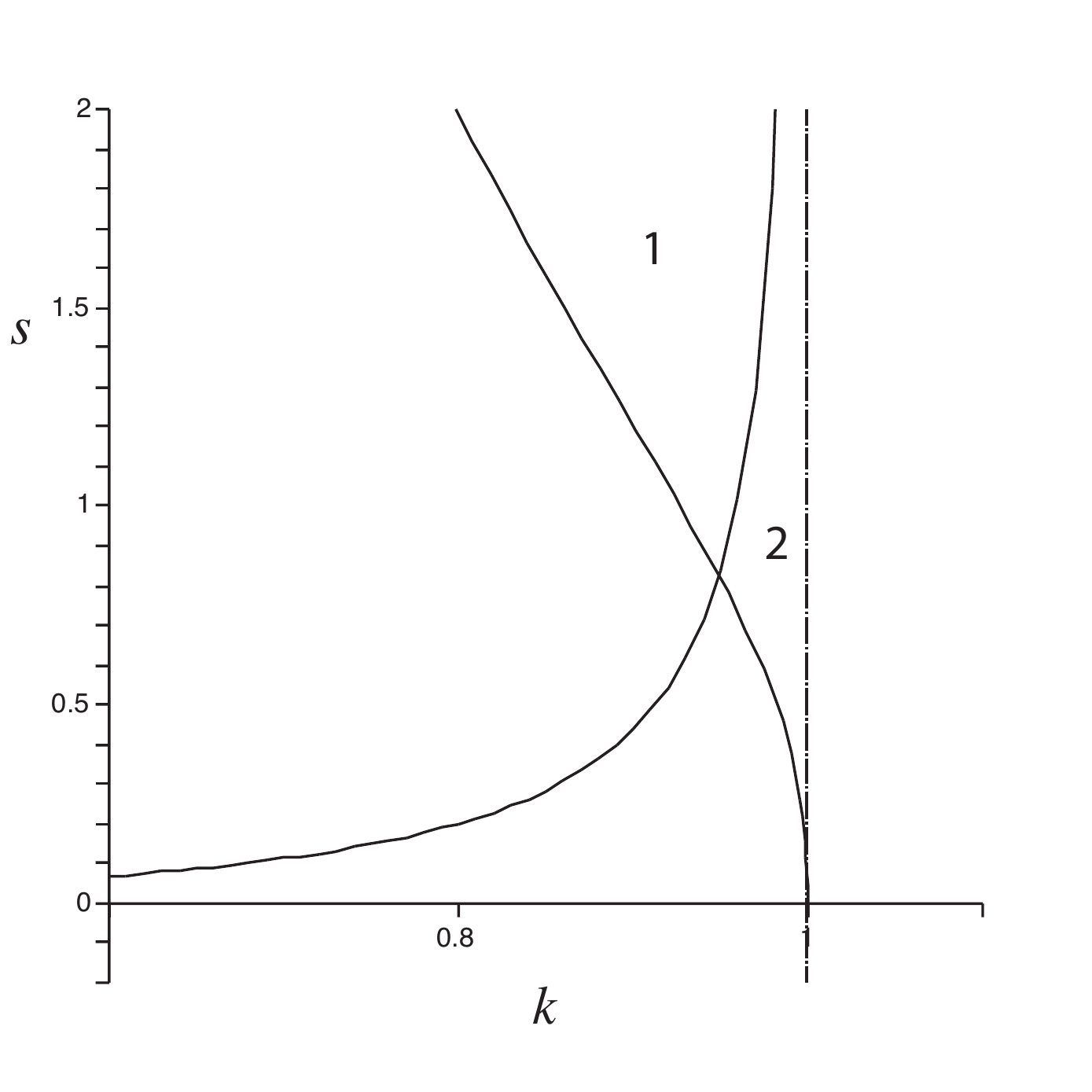}

\hfil (c) \hfill
\end{minipage}
\caption{In (a), region 1 (resp. 2) represents the subset of the $k$-$s$ space in which the inequality (\ref{condi1}) (resp. (\ref{condi2})) holds. 
%Both curves start at the origin and have an asymptote at $k=1$. 
In (b), region 1 is where the inequality in (\ref{req}) for the existence of the solution of CGL holds. The curve in (b) has an asymptote at $k=0$ and its  value at $k=1$ is $s=-1/4$.  
In (c), the top curve from (a) and the curve from (b) intersect near the point (0.949,0.824); regions 1 and 2 are explained in Figure \ref{fig7}}.\label{fig2}
\end{figure}

\setlength{\fred}{.3\textwidth}
\begin{figure}
\centering
%\hspace{1cm}(a)\hspace{-1cm}
%\scalebox{0.5}{\includegraphics{fig5.eps}}
%\hspace{-0cm}\hspace{1cm}(b)\hspace{-1cm}\scalebox{0.5}{\includegraphics{fig6.eps}}
%\hspace{-0cm}
\begin{minipage}{\fred}
\includegraphics[width=\fred]{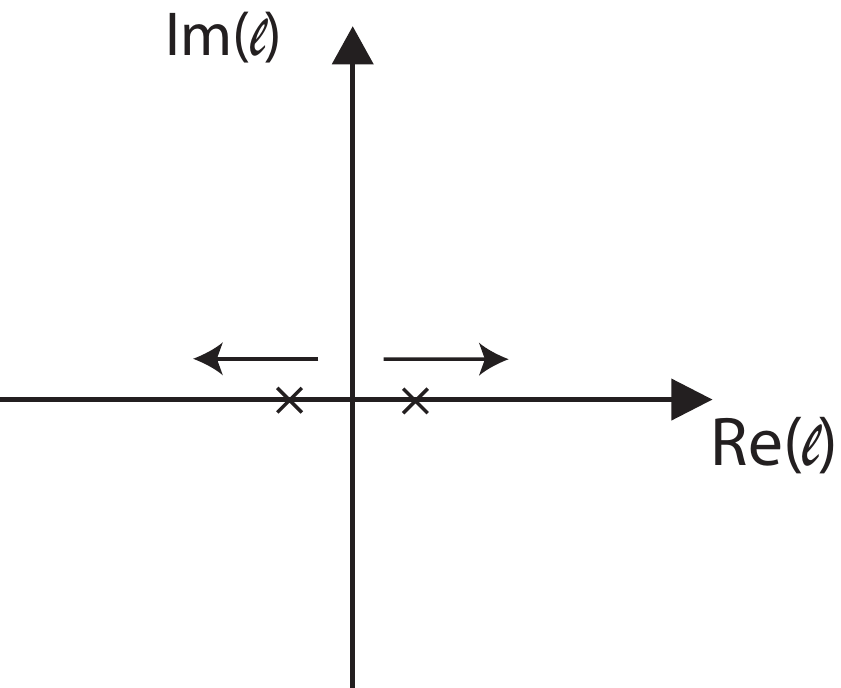}

\hfil (a) \hfill
\end{minipage}
\qquad
\begin{minipage}{\fred}
\includegraphics[width=\fred]{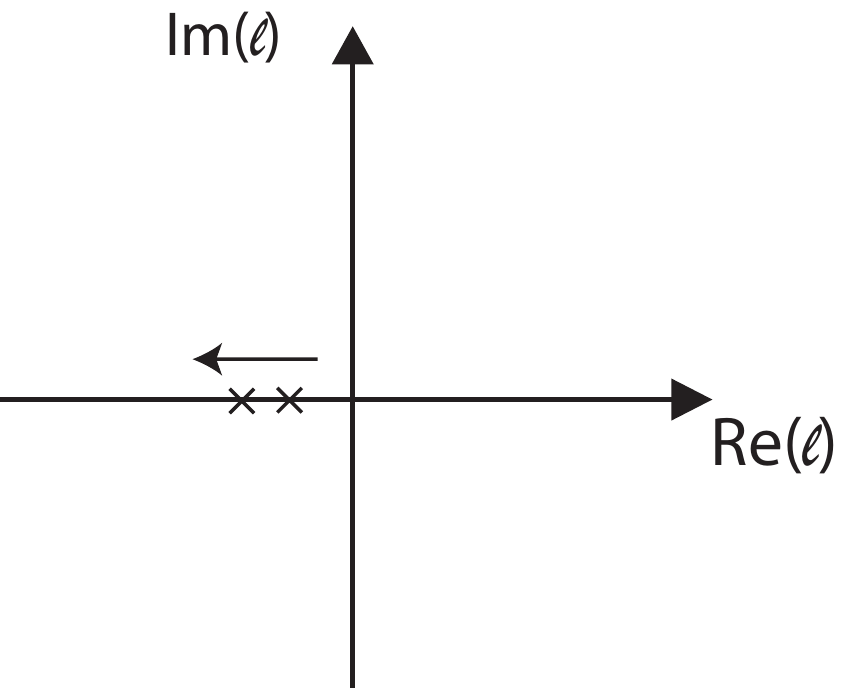}

\hfil (b) \hfill
\end{minipage}
\caption{Diagram (a) (resp. (b)) represents the behavior of the two emerging eigenvalues 
for values of $k$ and $s$ in region 1 (resp. 2) of Figure \ref{fig2}c . \label{fig7}}
\end{figure}

%\section{Conclusion}

%In this article, we have established the spectrum of the linearized NLS about the cnoidal wave solution (\ref{U0def}). Our results are twofold. First we have established that the point spectrum is restricted to 
%the imaginary axis, showing that  (\ref{Udef}) is spectrally stable with respect to perturbations having the same periodicity. In addition, we have obtained the continuous spectrum and showed that it intersects the right side of the complex plane, showing that  cnoidal wave solution (\ref{Udef}) are unstable when considered on the whole line. On the one hand, it is consistent with the fact that complex double points in the spectrum  of the Lax pair usually give rise to instabilities \cite{ErFoMc90,FoLe86,LiMc94}. ({\bf{Tom: we should discuss these references}}). On the other hand, this 
%is consistent with a recent work \cite{CaDe06} in which (\ref{U0def}) is numerically shown to be unstable.
%The results concerning the spectrum of linearized NLS are summed up in Theorem \ref{the2} and illustrated in Figure \ref{fig3}.

%Furthermore, we considered the persisting solution of CGL given in (\ref{Udef}). We used the Evans function technique and the Fredholm alternative to find criteria for eigenvalues of the linearized CGL to
%emerge from the origin to the right side of the complex plane when NLS is being perturbed to CGL. The results concerning this are illustrated in Figures \ref{fig2} and \ref{fig7}.

%

\appendix

\section{Expressions for ${\rm\mathbf{V}}^h_3$, ${\rm\mathbf{V}}^h_4$, ${\rm\mathbf{w}}^1_3$, ${\rm\mathbf{w}}^1_4$,
${\rm\mathbf{w}}^2_1$, and ${\rm\mathbf{w}}^2_2$}
\label{app:B}

The first component of ${\rm\mathbf{V}}^h_{3}$ is given by
$$
%\begin{aligned}
\left({\rm\mathbf{V}}^h_{3}\right)_1
%&=
%\dfrac{-\delta k}{k'^2 K}
%[(k'^2 K + (2k^2-1)E)\dn(\delta x) \sn(\delta x)\,\delta x
%+ (2k^2-1) K \dn(\delta x) \sn(\delta x) Z(\delta x)
%\\
%&+ K(2k^2 \dn^2(\delta x) -1)\cn(\delta x)]\hfill
 =-\delta k \left[ k'^2 \dn(\delta x) \sn(\delta x)\,\delta x
+ (2k^2-1) \left(\dfrac{E}{K}\delta x + Z(\delta x)\right)\right]
%\end{aligned}
$$
Furthermore, $\left({\rm\mathbf{V}}^h_{3}\right)_2=\left({\rm\mathbf{V}}^h_{3}\right)_1$ and $\left({\rm\mathbf{V}}^h_{3}\right)_i,i=3,4$ are both given by the derivative of $\left({\rm\mathbf{V}}^h_{3}\right)_1$.

The first component of ${\rm\mathbf{V}}^h_{4}$ is given by
$$\left({\rm\mathbf{V}}^h_{4}\right)_1
%=\dfrac{k}{(k^2-1)K}\left[ \delta(E + (k^2-1)K)\cn(\delta x)\,x
%+K(\cn(\delta x)Z(\delta x) - \dn(\delta x)\sn(\delta x))\right]
=-\dfrac{k}{k'^2} \left[ \left(\dfrac{E}{K} - k'^2\right) \cn(\delta x)\, \delta x
+ \cn(\delta x)Z(\delta x) - \dn(\delta x)\sn(\delta x) \right]
$$
Furthermore, $\left({\rm\mathbf{V}}^h_{4}\right)_2=-\left({\rm\mathbf{V}}^h_{4}\right)_1$ and $\left({\rm\mathbf{V}}^h_{4}\right)_i,i=3,4$ are given, respectively, by the derivatives of $\left({\rm\mathbf{V}}^h_{4}\right)_i,i=1,2$.

The expressions for ${\rm\mathbf{w}}^1_3$, ${\rm\mathbf{w}}^1_4$,
${\rm\mathbf{w}}^2_1$, and ${\rm\mathbf{w}}^2_2$ are rather complicated. However, for the purpose
of calculating the derivatives of the Evans functions, we only need their expressions evaluated at $x=T=4K/\delta$:
$$
{\rm\mathbf{w}}^1_3(T)=
\displaystyle{
\frac{-2\ri k\,E}{k'^2}
 }\begin{bmatrix}
\displaystyle{ \frac{4k^2E}{\delta k'^2}}\\\\
\displaystyle{-\frac{4k^2E}{\delta k'^2}}\\\\1\\\\-1
 \end{bmatrix},\;\;
 {\rm\mathbf{w}}^1_4(T)=
\displaystyle{
\frac{2\ri k\,E}{\delta k'^2}
 }\begin{bmatrix}
\delta^{-1}\\\\
\delta^{-1}\\\\
\displaystyle{\frac{4k^2E}{k'^2}}\\\\
\displaystyle{\frac{4k^2E}{k'^2}}
\end{bmatrix},
$$
$$
 {\rm\mathbf{w}}^2_1(T)=
\displaystyle{
\frac{-1}{\delta^2 k \,k'^2}
 }\begin{bmatrix}
\displaystyle{E(2k^2-1)+K k'^2}\\\\
\displaystyle{E(2k^2-1)+K k'^2}\\\\
\displaystyle{\frac{2\delta\left(
k'^4 {\it K}^{2}
+2\,k'^2 \left( 2\,{k}
^{2}-1 \right) {\it E}\,  {\it K}
-\left( 3\,{k}^{2}-4\,{k}^{4}-1
 \right)  {\it E}^{2}
  \right)}{k'^2}}
  \\\\
\displaystyle{\frac{2\delta\left(
k'^4 {\it K}^{2}
+2\,k'^2 \left( 2\,{k}
^{2}-1 \right) {\it E}\,  {\it K}
-\left( 3\,{k}^{2}-4\,{k}^{4}-1
 \right)  {\it E}^{2}
  \right)}{k'^2}}
   \end{bmatrix},
 $$
 $$
 {\rm\mathbf{w}}^2_2(T)=
\displaystyle{
\frac{-1}{\delta^2k \,k'^2}
 }\begin{bmatrix}
-\displaystyle{\frac{
2\left( k'^4 {\it K}^{2}
-2\,  k'^2 {\it E}\, {\it K}
 + \left( {k}^{2}+1 \right)  {\it E}^{2}\right)}
{\delta k'^2}}
\\\\
\displaystyle{\frac{
2\left( k'^4 {\it K}^{2}
-2\,  k'^2 {\it E}\, {\it K}
 + \left( {k}^{2}+1 \right)  {\it E}^{2}\right)}
{\delta k'^2}}
\\\\
 \displaystyle{K k'^2-E}\\\\
\displaystyle{-K k'^2+E}
 \end{bmatrix}.
$$

\section{The Evans Function for Cnoidal NLS Solutions}
\label{app:nls}
In this appendix we use solutions of the AKNS spectral problem for
focusing NLS solution
\begin{equation}\label{qzerocn}
q_0(x,t) = \delta k e^{-\ri \alpha t}\cn(\delta x;k),\qquad \alpha = \delta^2(1-2k^2)
\end{equation}
(corresponding to \eqref{U0def})
to construct solutions of the ODE system \eqref{fourbysys0}
using squared AKNS eigenfunctions.  We demonstrate that for all
but five values of $\ell$ along the imaginary axis (including $\ell=0$), such solutions form a basis for the solution space of \eqref{fourbysys0}.
Using this information, we determine all zeros of the Evans function
(other than possibly the four other imaginary values) 
using the Floquet spectrum of $q_0$.

\subsection{Squared Eigenfunctions and the Evans Ansatz}
In this subsection, we will assume that $q$ is a genus one finite-gap
solution of focusing NLS \eqref{NLS}.  These solutions have the form
$$q(x,t) = e^{-\ri\alpha t} U_0(\xi),\qquad \quad \xi = x - c t,$$
for real constants $\alpha,c$ and a $T$-periodic function $U_0$
which is not necessarily real.  (We will later specialize
to the case of the cnoidal solution above.)  The change of variables
$u_1(\xi,t) = e^{\ri \alpha t} q_1(x,t)$,
$u_2(\xi,t) = e^{-\ri\alpha t} q_2(x,t)$ will allow us to linearize about
a stationary periodic solution when it is applied to system \eqref{NLS},
yielding
\begin{equation}\label{genuNLS}
\begin{aligned}
\ri u_{1t}+u_{1\xi\xi}-\ri c u_{1\xi}+\alpha u_1+2u_1^2 u_2& =0,\\
-\ri u_{2t}+u_{2\xi\xi}+\ri c u_{2\xi}+\alpha u_2+2u_1 u_2^2& =0.
\end{aligned}
\end{equation}
(Note that, here, the time derivative holds $\xi$ fixed.)
We linearize this system about the solution $u_1 = \overline{u_2} = U_0(\xi)$, yielding
\begin{equation}\label{genvLNS}
\begin{aligned}
\ri \rv_{1t} + \rv_{1\xi\xi} -\ri c \rv_{1\xi} +\alpha \rv_1 + 4|U_0|^2 \rv_1 +2 U_0^2 \rv_2 &=0, \\
-\ri\rv_{2t} + \rv_{2\xi\xi} +\ri c \rv_{2\xi} +\alpha \rv_2
+ 4|U_0|^2 \rv_2 +2 \overline{U_0}^2 \rv_1 &= 0.
\end{aligned}
\end{equation}
We will now show how to obtain solutions for this linearized system, satisfying the ansatz
\begin{equation}\label{genansatz}
\rv_1 = \re^{\ell t} w_1(\xi), \qquad \rv_2 = \re^{\ell t} w_2(\xi).
\end{equation}

Recall the AKNS system for focusing NLS \cite{AKNS, BBEIM}:
\begin{equation}\label{AKNS}
\vpsi_x = \begin{bmatrix}
-\ri \lambda & \ri q\\
\ri \overline{q} & \ri \lambda
\end{bmatrix} \vpsi,
\quad
\vpsi_t = \begin{bmatrix}
\ri (|q|^2 - 2\lambda^2) & 2\ri \lambda q - q_x \\
2\ri \lambda\overline{q} + \overline{q_x} & \ri(2\lambda^2 - |q|^2)
\end{bmatrix}\vpsi.
\end{equation}
As is well known \cite{FoLe86, MOv},
the squares of the components of $\vpsi$ give a solution
of the linearization of \eqref{NLS} at $q$.
This linearization is
\begin{equation}\label{LNS}
\begin{aligned}
\ri g_t  + g_{xx} + 4 |q|^2 g + 2q^2 h &=0, \\
-\ri h_t  + h_{xx} +4 |q|^2 h + 2\overline{q}^2 g &= 0,
\end{aligned}
\end{equation}
and the solution is given by $g=\psi_1^2$, $h=\psi_2^2$.
It is easy to check that the substitutions
\begin{equation}\label{gensubs}
q(x,t) = \re^{-\ri\alpha t} U_0(\xi), \quad g(x,t)=\re^{-\ri \alpha t}\rv_1(\xi,t),
\quad h(x,t)=\re^{\ri \alpha t} \rv_2(\xi,t)
\end{equation}
turn \eqref{LNS} into \eqref{genvLNS}.

It is also well known that one can use Riemann theta functions
to produce formulas for finite-gap solutions of NLS and
the corresponding solutions of \eqref{AKNS}, starting
with a hyperelliptic Riemann surface $\Sigma$ of genus $g$ and certain other data
(see \cite{BBEIM} or \cite{CI05}).
The solution \eqref{qzerocn} arises for $g=1$, with the four branch points
of $\Sigma$ in conjugate pairs in the complex plane, so that
the equation of $\Sigma$ is
\begin{equation}\label{Sigmaform}
\mu^2 = (\lambda-\lambda_1)(\lambda - \overline{\lambda_1})
(\lambda-\lambda_2)(\lambda-\overline{\lambda_2}),
\end{equation}
where $\realpart \lambda_1 < \realpart \lambda_2$.
(See \S3 of \cite{CI05} for a derivation of these solutions from
general finite-gap formulas.)
In genus one, the finite-gap solutions of \eqref{AKNS}, which
are known as {\em Baker eigenfunctions}, take the form
\begin{align*}
\psi_1 &= \exp\left(\ri (\Omega_1(P)-\tfrac E2)x + \ri (\Omega_2(P)+\tfrac N2)t\right)
\Theta_1(\xi,P), \\
\psi_2 &= \exp\left(\ri (\Omega_1(P)+\tfrac E2)x + \ri (\Omega_2(P)-\tfrac N2)t\right)
\Theta_2(\xi,P),
\end{align*}
where $E,N$ are constants determined by $\Sigma$,
$P=(\lambda,\mu)$ is an arbitrary point on $\Sigma$
that projects to $\lambda\in \bC$,  $\Omega_{1,2}$ are
certain Abelian integrals on $\Sigma$, and $\Theta_{1,2}$
will be defined below.  For the moment, we note that
$\Theta_1,\Theta_2$ depend only on $P$ and $\xi = x - c t$ (where $c$ is determined
by the branch points of $\Sigma$)
 and are $T$-periodic in $\xi$.

After the substitutions \eqref{gensubs} the squared Baker eigenfunctions
yield solutions to \eqref{genvLNS} of the form
$$\begin{bmatrix} \rv_1 \\ \rv_2 \end{bmatrix} =
\begin{bmatrix} e^{\ri \alpha t} \psi_1^2 \\ e^{-\ri \alpha t} \psi_2^2 \end{bmatrix} =
e^{ 2\ri \xi \Omega_1(P) + 2\ri t\left(\Omega_2(P) + c\Omega_1(P)\right)}
\begin{bmatrix}
e^{\ri(N-cE+\alpha)t-\ri E\xi} \Theta_1(\xi,P)^2 \\
e^{\ri(cE-N-\alpha)t + \ri E\xi}\Theta_2(\xi,P)^2
\end{bmatrix}.$$
In particular, these $\rv_1,\rv_2$ satisfy the ansatz \eqref{genansatz}
for $\alpha = cE - N$ and
\begin{equation}\label{ellformula}
\ell = 2\ri(\Omega_2(P) + c\Omega_1(P)),
\end{equation}
with
\begin{equation}\label{genwform}
w_1(\xi)  = e^{\ri (2\Omega_1(P)-E)\xi} \Theta_1(\xi;P)^2, \quad
w_2(\xi) = e^{\ri (2\Omega_1(P)+E)\xi}\Theta_2(\xi;P)^2.
\end{equation}

\begin{prop}\label{elltomu}
The formula \eqref{ellformula} implies that $\ell = 4\ri \mu$.
\end{prop}
\begin{proof}
The surface $\Sigma$ has genus 1, and a basis $a$ and $b$
(shown in Figure \ref{g1basisfancy}) for homology cycles.
The differentials $\rd\Omega_1,\rd\Omega_2$ have
zero $a$-period.  Denoting their $b$-periods as $V,W$ respectively,
from \cite{CI05} we have
$$W/V= -c= (\lambda_1+ \lambda_2 + \overline{\lambda_1}+\overline{\lambda_2}).$$
(Note that the meaning of $c$ here is $-1$ times its definition in \cite{CI05}.)
Thus, the differential
$$\rd\ell = 2 \ri (\rd\Omega_2 + c\, \rd \Omega_1)$$
has zero $a$- and $b$-periods.  So, although $\Omega_{1,2}$ are
not well-defined on $\Sigma$, and so $\ell$ is a well-defined function on
$\Sigma$.

The integrals $\Omega_i$ use the branchpoint $\overline{\lambda_2}$ as
basepoint, so $\ell=0$ there.  Because $\interchange^*\rd\Omega_i = -\rd\Omega_i$
for $i=1,2$ (where $\interchange$ is the sheet exchange automorphism
$(\lambda,\mu) \mapsto (\lambda, -\mu)$), then $\interchange^* \ell = - \ell$.
So, it follows that $\ell$ vanishes at all four branch points, which are the
same points at which the coordinate $\mu$ vanishes.
The surface $\Sigma$ has two points at infinity, called $\infty_{\pm}$, where
both $\lambda$ and $\mu$ approach complex infinity and
$$\dfrac{\mu}{\lambda^{g+1}} = \dfrac{\mu}{\lambda^2}
= \pm \left(1 +\dfrac{c}2 \lambda^{-1}+ \calO(\lambda^{-2})\right)$$
respectively.
Combining the asymptotic expansions for $\Omega_1$ and $\Omega_2$ gives
$$\dfrac{\ell}{\lambda^2} = \pm 2\ri\left( 2 + c\lambda^{-1} + \calO(\lambda^{-2})
\right).$$
Therefore,
$$\ell/\mu = 4\ri + \calO(\lambda^{-2}).$$
Because $\ell/\mu$ is bounded and holomorphic on $\Sigma$, then it must be equal
to the constant $4\ri$, i.e.,
$$\ell = 4\ri \mu.$$
In other words, we have just proved the identity
\begin{equation}\label{muidentity}
\Omega_2(P) = -c \Omega_1(P) + 2 \mu
\end{equation}
for genus one NLS solutions.
\end{proof}

\subsection{Forming a Basis of Solutions}
Proposition \ref{elltomu} and \eqref{Sigmaform} imply that
for generic values of $\ell$ there are four distinct points $P \in\Sigma$
for which the Baker eigenfunctions may be used to construct a solution of
the eigenvalue problem \eqref{eigenvalueproblem}.  The exceptional values
are those for which two roots of \eqref{Sigmaform}, as a polynomial
equation for $\lambda$, coincide.

From here on, we specialize to the case where the branch points in \eqref{Sigmaform}
satisfy $\lambda_2 = -\overline{\lambda_1}$, with $\realpart \lambda_1 \ne 0$,
which yields the NLS solution \eqref{qzerocn}.  In this case,
$c=0$ and $E=V/2$, the branch points are related to
the elliptic modulus by $\lambda_1 / |\lambda_1| = -k' + \ri k$,
and the finite-gap solution coincides with \eqref{qzerocn} with 
$\delta = 2|\lambda_1|$.
The exceptional values of $\mu$ are
\begin{equation}\label{muexcept}
\mu =\pm |\lambda_1|^2, \qquad \mu = \pm \impart(\lambda_1^2),
\end{equation}
corresponding respectively to a double root occurring at $\lambda=0$,
or two pairs of roots coinciding
at opposite points along the real axis (if $k^2 < 1/2$) or imaginary axis (if $k^2 > 1/2$).

When $\mu$ is not an exceptional value, we will construct a matrix solution $W(\xi;\ell)$ of
\eqref{fourbysys0} where each column is of the form
$[w_1(\xi),w_2(\xi),w_1'(\xi),w_2'(\xi)]^T$ and $w_1,w_2$ are given by
\eqref{genwform} for one of the four points $P=(\lambda,\mu)$ satisfying \eqref{Sigmaform}
for the given $\mu$.  For purposes of showing that the matrix solution is nonsingular,
it suffices to evaluate its determinant at one value of $\xi$.

We now need to specify the form of the factors $\Theta_1,\Theta_2$ in \eqref{genwform}.
Specializing the formulas given in \S5 of \cite{CI05},  and using $\xi = x-ct = x$,
we have
$$\Theta_1(x;P) = \dfrac{\theta(A(P)+\ri V x - D)\theta(D)}
{\theta(\ri V x-D)\theta(A(P)-D)}, \quad
\Theta_2(x;P) = -\ri e^{\Omega_3(P)}\dfrac{\theta(A(P)+\ri V x - D-R)\theta(D)}
{\theta(\ri V x-D)\theta(A(P)-D)}.$$
In these formulas,
\begin{itemize}
\item $V = \pi |\lambda_1 -\overline{\lambda_2}|/K = 2\pi |\lambda_1|/K$;
\item $A(P)$ is the Abel map, obtained by integrating the differential
$\omega = V \, \rd\lambda/(2\mu)$ from $\infty_-$ to $P$ on $\Sigma$;%
%\footnote{To distinguish the points on $\Sigma$ lying above $\lambda = \infty$,
%we denote by $\infty_\pm$ those for which $\mu /\lambda^2$ tends to $\pm1$,
%respectively.}
\item $\Omega_3(P)$ is the integral, from $\overline{\lambda_2}$ to $P$,
of the unique meromorphic differential $\rd\Omega_3$ on $P$ which has zero
$a$-period and satisfies $\rd\lambda_1 \sim \pm \lambda^{-1}\rd\lambda$ near $\infty_\pm$;
\item $R$ is minus the period of $\rd\Omega_3$ on the cycle $b$ running
from $\lambda_1$ to $\lambda_2$ (see Figure \ref{g1basisfancy}), and is also
the value $A(\infty_+)$, when the
path of integration is chosen to avoid the homology cycles $a$ and $b$
(see \S5 in \cite{CI05}); in the
case when $c=0$, $R=\pi(K'/K + \ri)$ (see \S3 in \cite{CI05});
\item $D$ is an arbitrary constant which is pure imaginary;
\item $\theta(z)$ is the Riemann theta function with period $2\pi \ri$
and quasiperiod $B= - 2\pi K'/K$.  It is related to the
Jacobi theta functions of modulus $k$ by $\theta(z) = \theta_3(z/(2\ri))$.
\end{itemize}

\begin{figure}[h]
\centering
\includegraphics[height=2in]{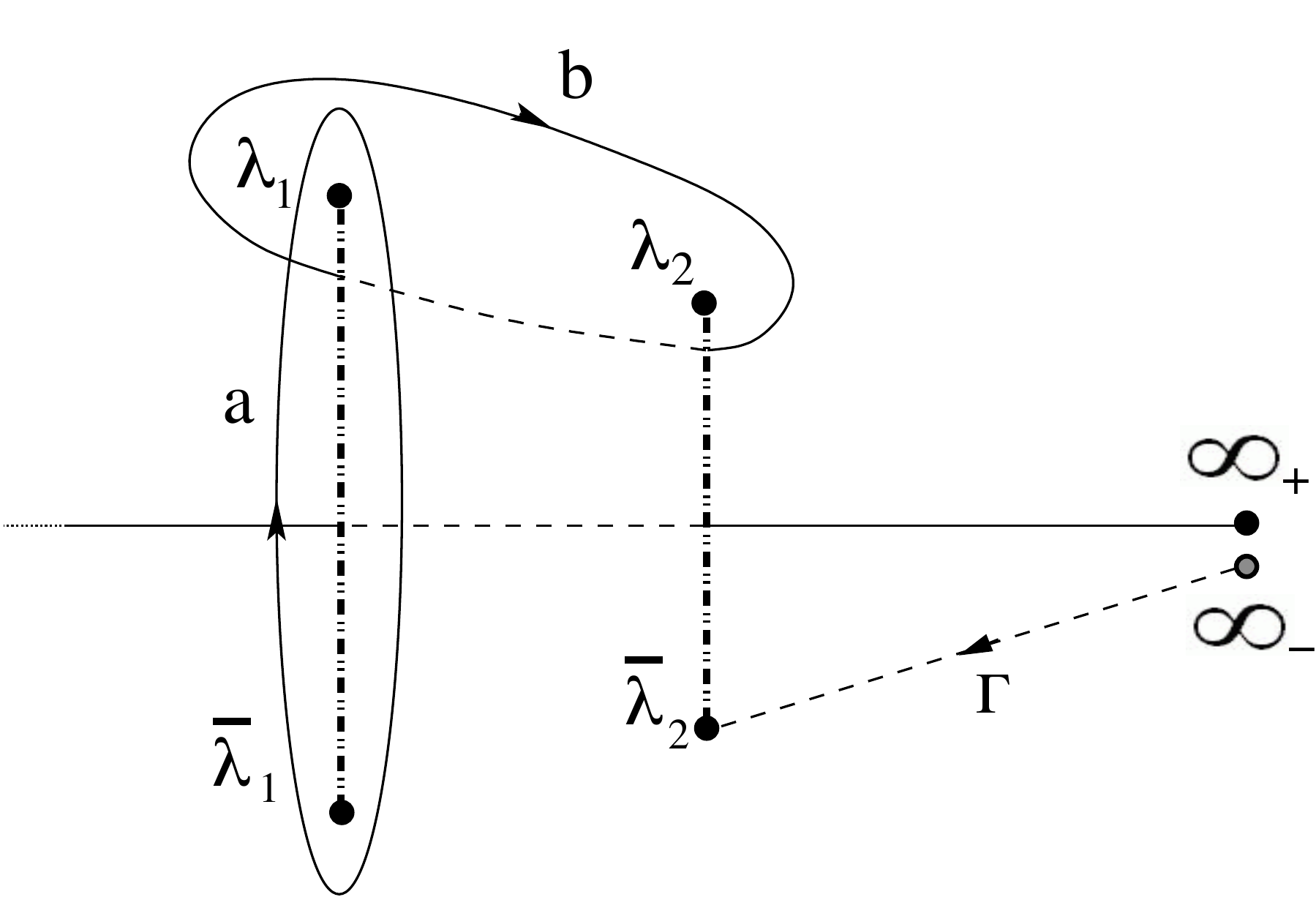}
\caption{Homology cycles and other integration paths on the
genus one Riemann surface with complex conjugate branch points.
Solid curves appear on the upper sheet, dashed curves on the
lower sheet, and branch cuts extend between each branch point
and its conjugate.
The horizontal line represents a set of points where
$\mu$ is real and positive.}\label{g1basisfancy}
\end{figure}

In genus one, the constant $D$ has no significance, as it can be absorbed
through a shift in $x$.  Thus, we may assume $D=0$ in the above formulas.
Once this is done, the columns of $W$ for $\xi=0$ take the form
\begin{equation}\label{winfgh}
\begin{bmatrix}w_1(0)\\ w_2(0)\\ w_1'(0) \\ w_2'(0) \end{bmatrix}
=
\begin{bmatrix} 1 \\ -f(P)^2 \\ 2\ri g(P) \\ -2\ri f(P)^2 h(P) \end{bmatrix},
\end{equation}
where
\begin{multline*}
f(P) := e^{\Omega_3(P)}\dfrac{\theta(A(P)-r)}{\theta(A(P))}, \quad 
g(P) := \Omega_1(P) +V \Theta(A(P))-\frac{E}2, \\
h(P) := \Omega_1(P) +V \Theta(A(P)-r)+\frac{E}2,\hspace{8cm}
\end{multline*}
and $\Theta(z) = \theta'(z)/\theta(z)$ is the logarithmic derivative
of the Riemann theta function.  (Note that $A(P)$ and $\Omega_i(P)$ are
not individually well-defined on $\Sigma$, because of
the nonzero periods of the corresponding differentials along the
homology cycles. We take the convention that
the paths of integration for these differ by a fixed path $\Gamma$ (shown
in Figure \ref{g1basisfancy}) from $\infty_-$ to
$\overline{\lambda_2}$ in $\Sigma_0$,
where $\Sigma_0$ denotes the simply-connected domain that
results from cutting $\Sigma$ along the homology cycles.
With this convention, $f,g,h$ are well-defined meromorphic functions on $\Sigma$.)

The Riemann surface $\Sigma$ has two holomorphic involutions,
namely $\interchange:(\lambda,\mu) \mapsto (\lambda,-\mu)$ and
$\sigma:(\lambda,\mu) \mapsto (-\lambda,\mu)$.  Because $\interchange^* \rd\Omega_i =
-\rd\Omega_i$, $\interchange^* \omega = -\omega$ and
$\interchange$ fixes the basepoint for $\Omega_i$, $i=1,2,3$,
it follows that
$$f(\interchange P) = 1/f(P),\qquad
g(\interchange P) = -h(P).$$
Because $\sigma^* \rd\Omega_1 = -\rd\Omega_1$,
%$\sigma^* d\Omega_2 = d\Omega_2$,
$\sigma^* \rd\Omega_3 = \rd\Omega_3-\omega$, $\sigma^* \omega = -\omega$,
and $E=V/2$,
it follows that
$$f(\sigma P) = f(P), \qquad g(\sigma P) = -g(P).$$

Let $P_1, P_2, P_3, P_4$ be four points on $\Sigma$ corresponding to
a given (non-exceptional) value of $\mu$.  Without loss of generality,
we may assume that $P_2 = \sigma(P_1)$ and $P_4= \sigma(P_3)$.  Using the
above formulas for the behaviour of $f,g,h$ under $\sigma$, we find that
the matrix with columns given by \eqref{winfgh} is
$$W(0) = \begin{bmatrix}
1 & 1 & 1 & 1 \\
-f(P_1)^2 & -f(P_1)^2 & -f(P_3)^2 & -f(P_3)^2 \\
2\ri g(P_1) & -2\ri g(P_1) & 2\ri g(P_3) & -2\ri g(P_3) \\
-2\ri f(P_1)^2 h(P_1) & 2\ri f(P_1)^2 h(P_1)
& -2\ri f(P_3)^2 h(P_3) & 2\ri f(P_3)^2 h(P_3)
\end{bmatrix}.
$$

\begin{prop}\label{niftygh} The functions $g, h$ satisfy the identity
$$- g(P)= h(P) =\lambda,\qquad P=(\lambda,\mu)\in \Sigma.$$
\end{prop}
\begin{proof} The poles of $g(P) = \Omega_1(P) + \Theta(A(P))-E/2$ can only occur
where $\Omega_1$ has a pole or $\theta(A(P))=0$.  (We now work on the cut
Riemann surface $\Sigma_0$.)  These occur only
at $\infty_-$ (near which $\Omega_1 \sim -\lambda+\calO(1)$, but $\theta(A(P)) = \theta(0)=\sqrt{2K/\pi}\ne 0$),
and at $\infty_+$ (near which $\Omega_1 \sim \lambda+\calO(1)$ and
$\theta(A(P))=\theta(R)=0$).  Near $\infty_+$, use $w=\lambda^{-1}$ as local coordinate.
Then $A(P) = \tfrac{V}2 \int_{\infty_-}^P \mu^{-1}\rd\lambda$ implies that
$$\dfrac{\rd A(P)}{\rd w} = -\dfrac{V}{2 w^2 \mu} =
-\dfrac{V}2 (1 + \calO(w^2))$$
and therefore $A(P) = R - \tfrac{V}{2} w + \calO(w^3)$.  Because $\theta(z)$
has a simple zero at $z=R$,
$$g(P) = \Omega_1(P) +V\dfrac{\theta'(A(P))}{\theta(A(P))} - \frac{E}2
= w^{-1} + V\left(\dfrac{-2}{V}\right) w^{-1} + \calO(1) = -\lambda + \calO(1).$$

The oddness of $g$ with respect to $\sigma$ implies that
$g(P)=0$ when $P$ lies over the origin in the $\lambda$-plane.
Therefore, the quotient $g(P)/\lambda$ is a bounded holomorphic
function on $\Sigma$, and so must be a constant.  The above asymptotic expansion
for $g$ implies that $g(P) = -\lambda$, and it follows from the
fact that $h(\interchange P) = -g(P)$ that $h(P) = \lambda$.
\end{proof}

Taking Prop. \ref{niftygh} into account, we compute
\begin{equation}\label{detN}
\det  W(0) = -16 \lambda(P_1)\lambda(P_3) (f(P_1)^2 - f(P_3)^2)^2.
\end{equation}

\begin{prop}\label{nonsingN}
For $\mu$ not equal to zero or any of the
exceptional values in \eqref{muexcept}, the matrix $N$ is nonsingular.
{\rm (We conjecture that the analogous statement is true for genus one finite-gap
solutions in general.)}
\end{prop}

\begin{proof} Because we are excluding the exceptional values, neither of
the $\lambda$-values in \eqref{detN} are zero, and it only remains to
establish that $f(P_1) \ne \pm f(P_3)$.

Define $F(P) = f(P) + 1/f(P)$, which is a well-defined
meromorphic function of $\lambda$.  The
formula for $f(P)$ shows that its only pole is a second-order pole at $\infty_+$
(because $\exp(\Omega_3(P)) \sim \lambda$ near there),
and hence $F$ is a second-order polynomial in $\lambda$.  It is
known that $f(\pm \lambda_1) = 1$ and $f(\pm \overline{\lambda_1}) =-1$ (see \S7 of \cite{isopaper}).
Hence,
$$F(\lambda) = \dfrac{4(\lambda^2 - d)}{\lambda_1^2 - \overline{\lambda_1}^2},
\qquad d:=\realpart(\lambda_1^2).$$
Using the quadratic formula, and the fact that $f$ has a second-order zero at $\infty_-$,
we obtain
\begin{equation}\label{formulaforf}
f(P) = \dfrac{2(\lambda^2 - d + \mu)}{\lambda_1^2 - \overline{\lambda_1}^2}.
\end{equation}

Let $y_1=\lambda(P_1)$ and $y_3 = \lambda(P_3)$.  Because $y_1$ and $y_3$ are
distinct roots of 
\begin{equation}\label{symmetricsigmaq}
(\lambda^2 - \lambda_1^2)(\lambda^2 -\overline{\lambda_1}^2)=\mu^2,
\end{equation}
and $y_1 \ne -y_3$,
then $y_1^2 + y_3^2 = 2d$.
Hence, $F(y_1)=2(y_1^2-y_3^2) = -F(y_3)$, and therefore $f(P_1) \ne f(P_3)$.
Furthermore, using \eqref{formulaforf}, we see that
$$(\lambda_1^2 - \overline{\lambda_1}^2)( f(P_1)+f(P_3)) = 2(y_1^2+y_3^2 - 2d+2\mu)=4\mu.$$
Thus, because we assume $\mu\ne 0$, $f(P_1) \ne -f(P_3)$.
\end{proof}

Assuming that $\mu$ is not zero or one of the exceptional values,
we can define the transfer matrix
$$N = W(T) W(0)^{-1},$$
whose eigenvalues describe the growth of solutions to \eqref{fourbysys0} over one
period.  (In particular, there is a periodic solution if and only if $N$ has eigenvalue one.)
Using the fact that $T = 4\pi/V=2\pi/E$ and $\Omega_1(\sigma \, P) = -\Omega_1(P) + E/2$ on the cut surface, we calculate that
$$W(T) = W(0) \begin{bmatrix}
 e^{2\ri \Omega_1(P_1) T} &0&0&0 \\
 0 & e^{-2\ri\Omega_1(P_1)T} & 0 & 0 \\
 0 & 0 & e^{2\ri \Omega_1(P_3)T} & 0 \\
 0 & 0 & 0 & e^{-2\ri \Omega_1(P_3)T}
 \end{bmatrix}.
$$
Then the characteristic polynomial of the transfer matrix is
$$\det(N - \tau I) = ((\tau+1)^2 - 4\tau\cos(\Omega_1(P)T))((\tau+1)^2 - 4\tau\cos(\Omega_1(P_3)T)).$$
By substituting $\tau=1$ in this formula, we obtain the
\begin{prop}\label{Wperiodic}
For $\mu$ not equal to zero or the
exceptional values in \eqref{muexcept}, the system \eqref{fourbysys0},
where $\ell = 4\ri \mu$, has a periodic solution if and only if
$\sin(\Omega_1(P_1)T)=0$ or $\sin(\Omega_1(P_3)T) = 0$.
\end{prop}

\subsection{Imaginary Zeros of the Evans Function}
For periodic finite-gap solutions of NLS, the {\em Floquet discriminant} \cite{CI05}
takes the form
$$\Delta(\lambda) = 2 \cos\left(\dfrac{2\pi}{V}\Omega_1(P)\right)=
2\cos\left(\dfrac{T}{2} \Omega_1(P)\right).$$
(While $\Omega_1$ is odd with respect to
the involution $\interchange$, the even-ness of the cosine makes $\Delta$
a well-defined function of $\lambda$.)
This has the property that the AKNS system \eqref{AKNS} admits a T-periodic (respectively,
antiperiodic) solution if and only $\Delta = \pm 2$.

The consequence of Propositions \ref{nonsingN} and \ref{Wperiodic} is that, for values of $\ell$ that correspond to nonzero values of $\mu$
excluding those in \eqref{muexcept},
the Evans function is equal to zero if and only if $\Delta=\pm 2$ or $\Delta = 0$ at
an opposite pair out of the four $\lambda$-values corresponding to $\mu$.  When this
happens, the corresponding pair of columns of $W$ are periodic, and
the corresponding value of $\ell$ is a zero of the Evans function of
geometric multiplicity two.
Thus, we can use the Floquet discriminant to find the (nonexceptional)
zeros of the Evans function.\footnote{Even if one of the double
points gives an exceptional value of $\mu$, it still gives a zero of the Evans function.
However, there may be additional zeros at exceptional values of $\mu$, corresponding
to possible periodic solutions of \eqref{fourbysys0} obtained from
solutions of the form \eqref{genwform} by reduction of order.}

In \S 3.5 of \cite{CI05}, the discriminant for genus one finite-gap
solutions is calculated as
\begin{equation}\label{deltaformula}
\Delta(\lambda)=-2\cos\left[2\ri K\left( Z(u) - \beta^2 \dfrac{\cn u \dn u \sn u}{1-\beta^2\sn^2 u}\right)\right],
\end{equation}
where the variable $u$ is related to $\lambda$ by
\begin{equation}\label{snoo}
\sn^2 u = \varphi(\lambda) := \dfrac{(\lambda_1-\blamda_2)(\lambda-\blamda_1)}{(\lambda_1-\blamda_1)(\lambda-\blamda_2)}
\end{equation}
and the parameter $\beta$ and the modulus $k$ are related to the branch points  by
$$k^2 = \dfrac{-(\lambda_1-\blamda_1)(\lambda_2 - \blamda_2)}{|\lambda_1-\blamda_2|^2}, \quad
\beta^2 = \dfrac{\lambda_1 -\blamda_1}{\lambda_1-\blamda_2}.$$

In the symmetric case (i.e., $\lambda_2 = -\blamda_1$) which we are considering here, these specialize to
$$-k' + \ri k = \lambda_1/|\lambda_1|, \qquad \beta^2 = k(k-\ri k'),$$
and the variable $u$ is determined by $\lambda$ as follows.  The function $\sn^2 u$ has periods
$2K$ and $2\ri K'$, and is even in $u$, so it suffices to restrict $u$ to the set
$(0,K] \times (-K',K']$ in the complex plane.  Then there is a unique $u$
in this set for each $\lambda$ in the extended complex plane.  Conversely,
solving \eqref{snoo} for $\lambda$ in terms of $u$ gives
$$\lambda(u) = \dfrac{\lambda_1 \sn^2 u + \blamda_1\beta^{-2}}
{\beta^{-2}-\sn^2 u}.$$

The locus of $\lambda$-values for which $\Delta(\lambda)$ is real
and between $2$ and $-2$ is known as the {\em continuous Floquet spectrum} of
the NLS potential $q$.  For genus one solutions, it consists of the
real axis and two bands terminating at the branch points.  (For
low values of $k$ the bands emerge from the real axis, but
for $k$ sufficiently near 1, the bands become detached
from the real axis (see Figure \ref{spectrumfig}); the
transition occurs around $k \simeq .9089$ \cite{IS}.)

The periodic points (i.e., where $\Delta=\pm 2$) consist of the branch points and
countably many points (known as double points) on the real axis.
Points where $\Delta=0$ will naturally interlace the periodic points along
the real axis, but will also occur along the imaginary axis if
the bands are detached.  When this happens, though,
the corresponding value of $\mu$ is real.  Hence, all the
zeros of the Evans function will occur along the imaginary
axis in the complex $\ell$-plane.

Figures \ref{spectrumfig} and \ref{fig3}
show the level sets of $\Delta(\lambda)$ for several different moduli,
and the zeros of Evans function in the $\ell$-plane.
The latter figures also include the locus of $\ell$-values which are related
by \eqref{symmetricsigmaq}, with $\ell =4\ri \mu$, 
to $\lambda$-values in the continuous spectrum.

\begin{figure}[ht]
\centering
\includegraphics[width=\textwidth]{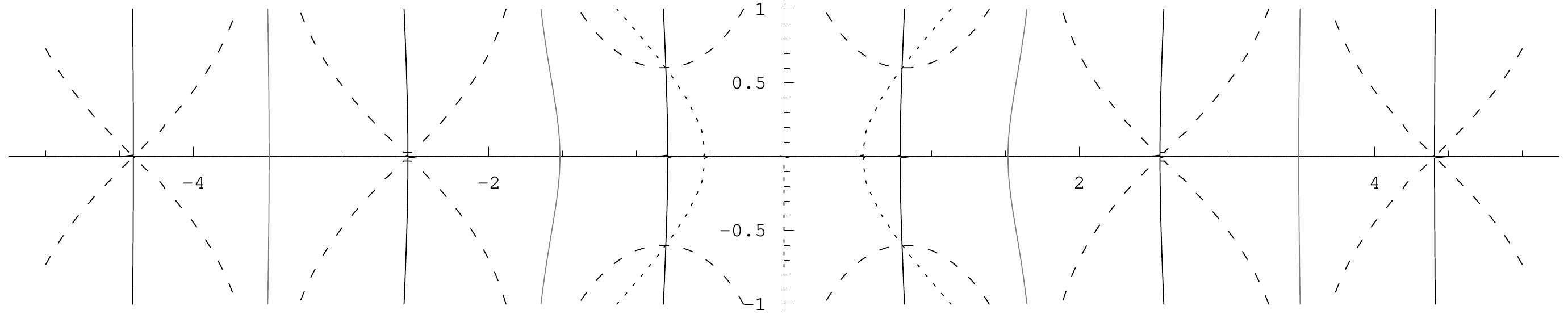}

\bigskip

\includegraphics[width=\textwidth]{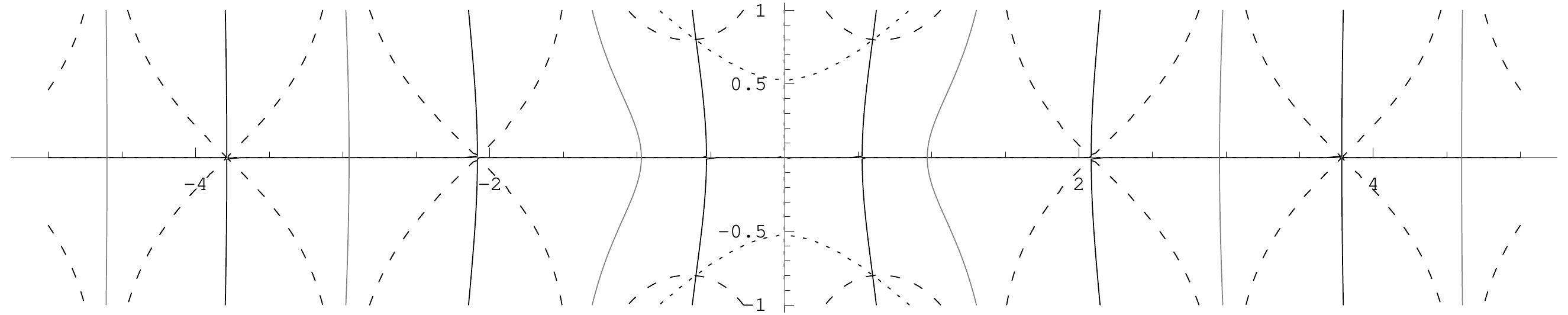}

\bigskip

\includegraphics[width=.7\textwidth]{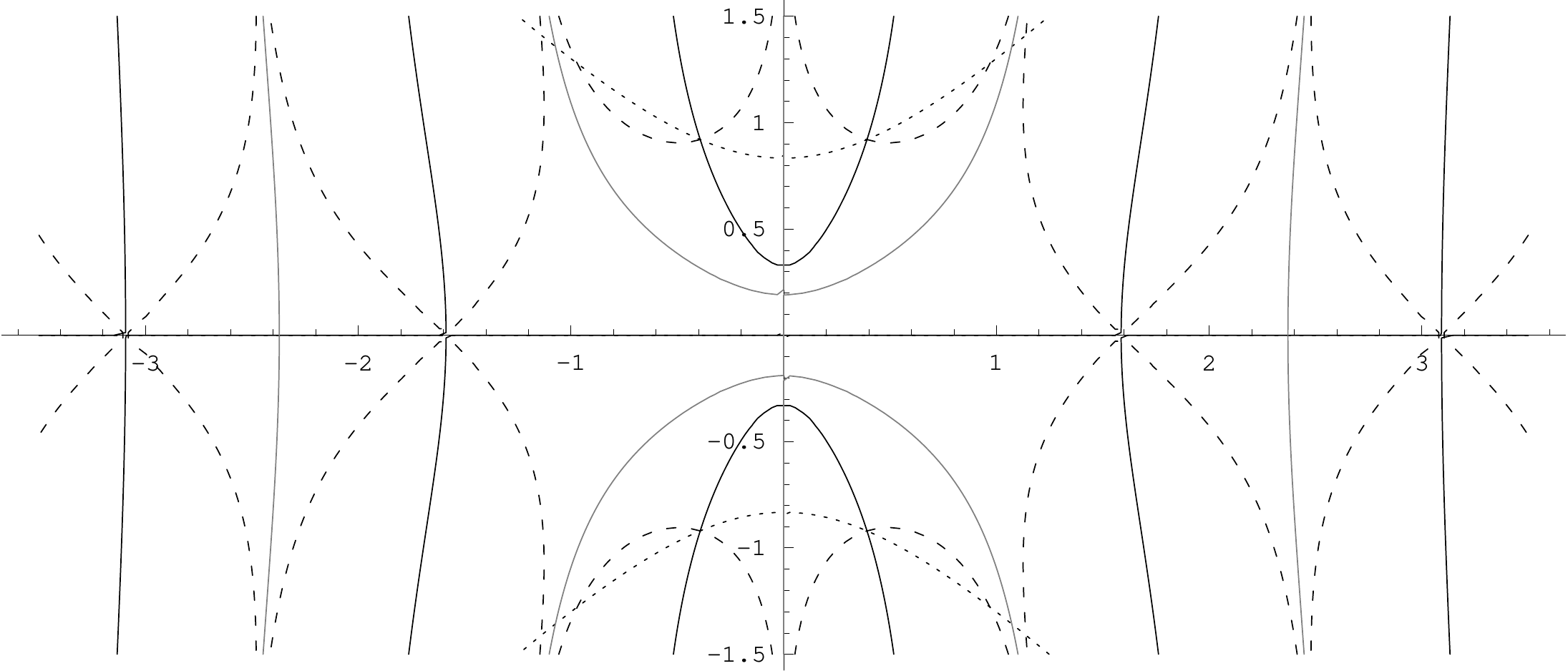}

\caption{Level curves of the Floquet discriminant $\Delta$ in the
complex $\lambda$-plane for elliptic moduli
$k=.6$, $k=.8$ and $k=.92$, respectively.  (We take $|\lambda_1|=\delta/2=1$.)  
In each of the diagrams, the black curves indicate where $\Delta$ is real,
the gray curves where $\Delta$ is imaginary, the dashed curves where
$\realpart \Delta =\pm2$, and the dotted curves where $\mu$ is real (which
includes the real and imaginary axes).  Thus, branch points occur
where dotted intersects dashed, periodic points occur where
black intersects dashed, and points where $\Delta=0$ occur
where black intersects grey.}\label{spectrumfig}
\end{figure}

\begin{figure}[ht]
\centering
\includegraphics[width=2in]{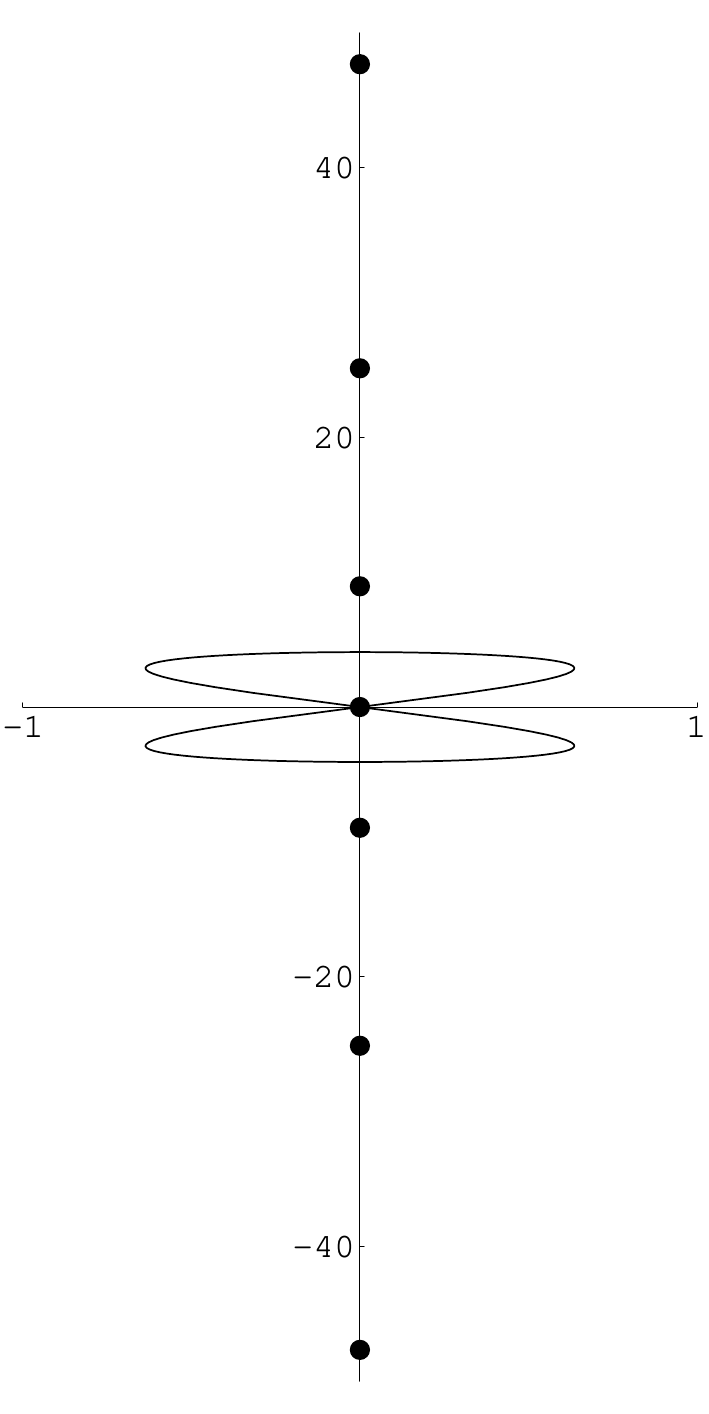}
\includegraphics[width=2in]{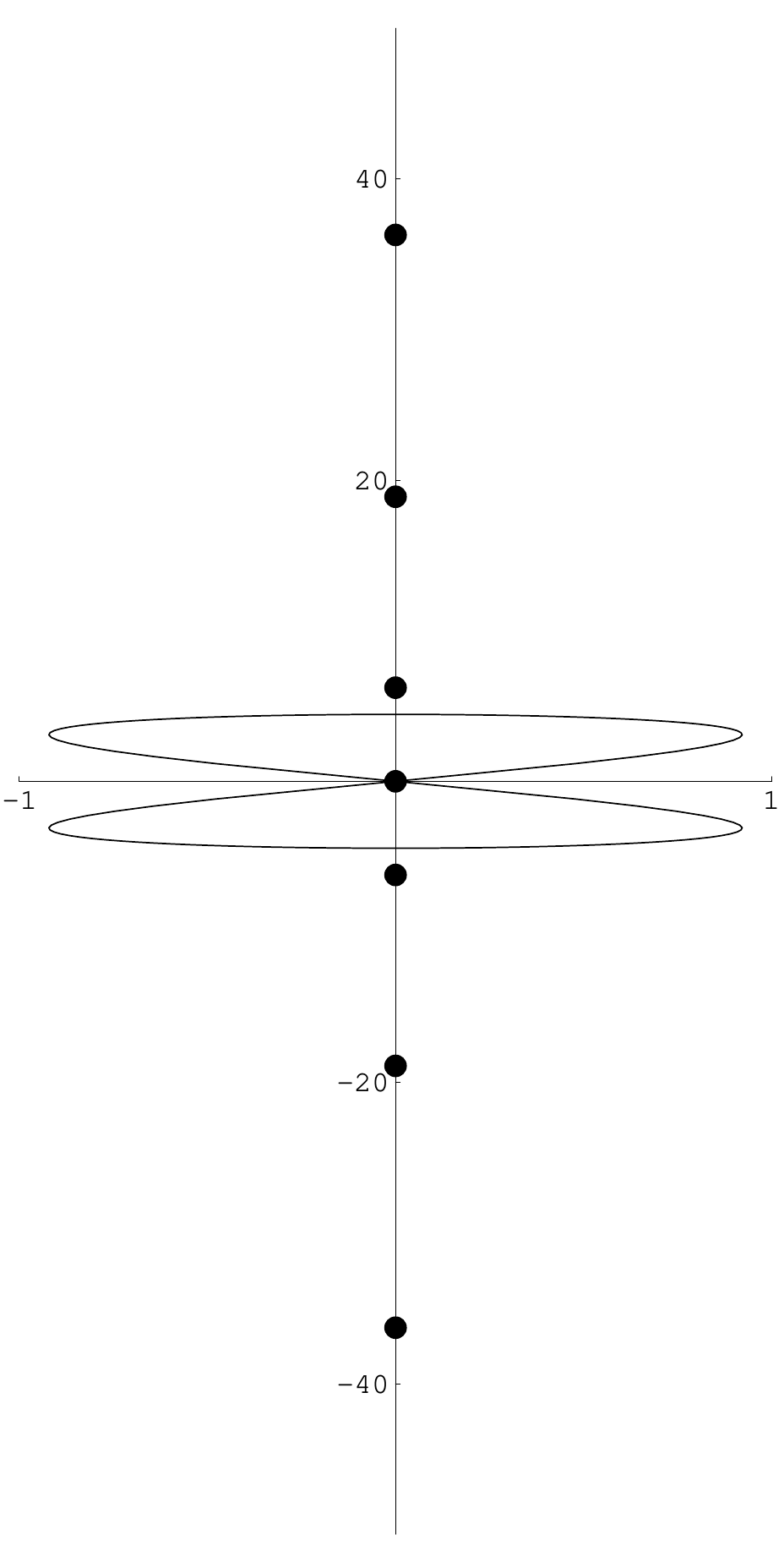}
\includegraphics[width=2in]{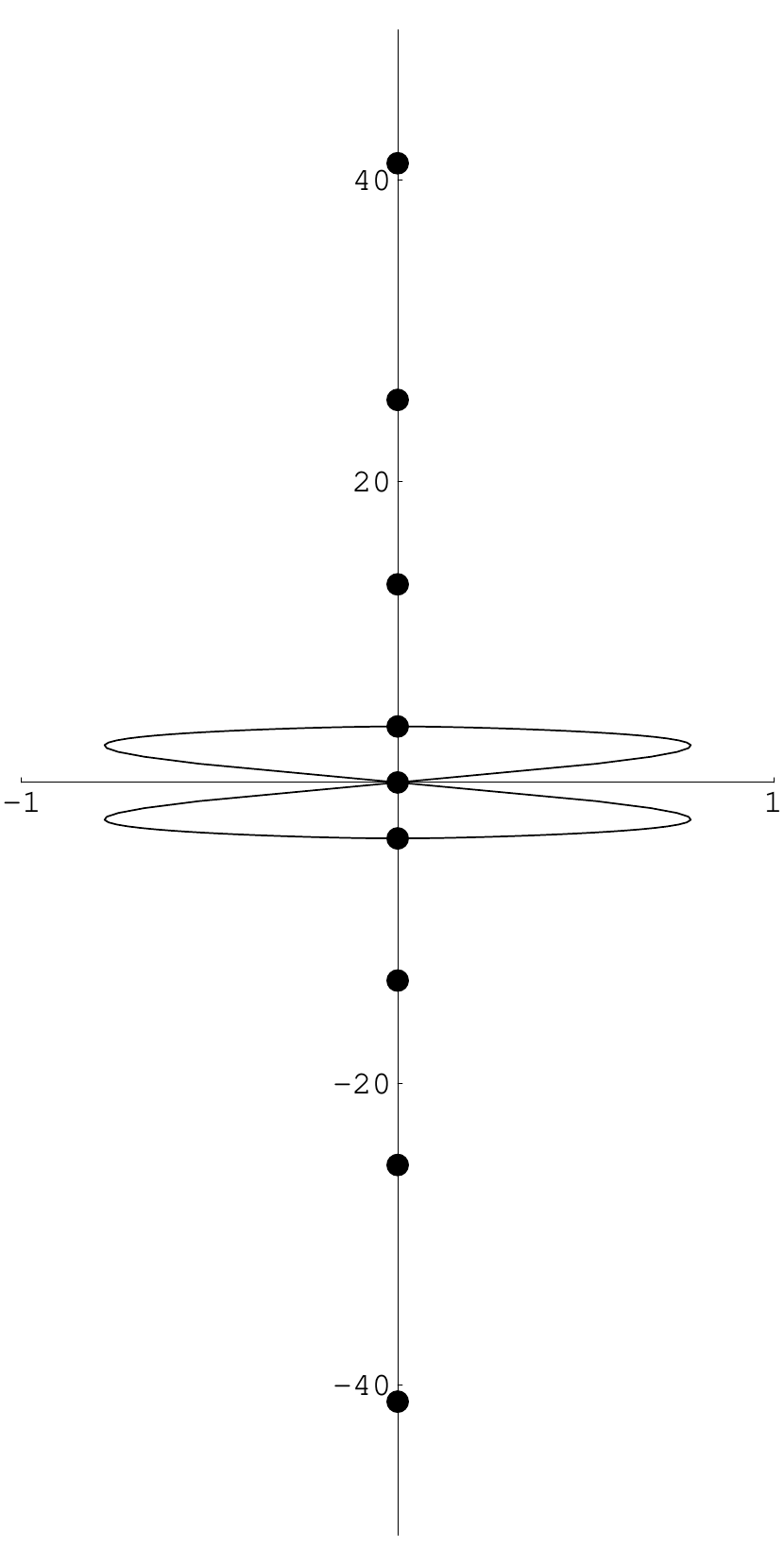}

\caption{Zeros of the Evans function in the complex $\ell$-plane,
for $k=.6$, $k=.8$ and $k=.92$, respectively.  Note that all zeros are
of multiplicity two, except for the origin, which is of multiplicity four.  The
continuous spectrum consists of the imaginary axis and the figure eight curve.}
\label{fig3}
\end{figure}

\end{document}